\documentclass[11pt]{article}
\usepackage{amsmath}  
\usepackage{amsfonts,dsfont}
\usepackage{amssymb}
\usepackage{amsthm}
\usepackage{mathtools}
\usepackage{thmtools}
\usepackage{bm} 
\usepackage{verbatim}
\usepackage{xcolor}
\usepackage{physics}

\usepackage{authblk}
\usepackage{graphicx}
\usepackage[hidelinks]{hyperref}
\usepackage[english]{babel}
\usepackage{float}

\newlength{\bredde}
\def\slash#1{\settowidth{\bredde}{$#1$}\ifmmode\,\raisebox{.15ex}{/}
\hspace*{-\bredde} #1\else$\,\raisebox{.15ex}{/}\hspace*{-\bredde} #1$\fi}
\newcommand{\diag}{{\rm diag\,}}

\newcommand{\eins}{\leavevmode\hbox{\small1\kern-3.8pt\normalsize1}}
\newcommand{\be}{\begin{equation}}
\newcommand{\ee}{\end{equation}}
\newcommand{\bee}{\begin{eqnarray}}
\newcommand{\eee}{\end{eqnarray}}

\newcommand{\GinUE}{\textup{GinUE}}
\newcommand{\GinOE}{\textup{GinOE}}

\newtheorem{thm}{Theorem}[section]

\newtheorem{cor}[thm]{Corollary}
\newtheorem{lem}[thm]{Lemma}
\newtheorem{prop}[thm]{Proposition}
\newtheorem*{prob*}{Problem}
\newtheorem*{thm*}{Theorem}

\theoremstyle{definition}
\newtheorem{defn}[thm]{Definition}

\newtheorem*{defn*}{Definition}
\newtheorem{rem}[thm]{Remark}

\newtheorem*{rem*}{Remark}
\numberwithin{equation}{section}
\newcommand{\erfc}{\text{erfc}}

\providecommand{\keywords}[1]
{
  \small	
  \textbf{Keywords:} #1
}

\title{Mean left-right eigenvector self-overlap in the real Ginibre ensemble}
\author{T. R. W\"urfel, M. J. Crumpton, Y. V. Fyodorov}
\affil{Department of Mathematics, King's College London, London WC2R 2LS, UK}

\date{\today}

\begin{document}

\maketitle
\begin{abstract}
\noindent
We study analytically the Chalker-Mehlig mean diagonal overlap $\mathcal{O}(z)$ between left and right eigenvectors associated with a complex eigenvalue $z$ of $N\times N$ matrices in the real Ginibre ensemble (GinOE). We first derive a general finite $N$ expression for the mean overlap and then investigate several scaling regimes in the limit $N\rightarrow \infty$. While in the generic spectral bulk and edge of the GinOE the limiting expressions for $\mathcal{O}(z)$ are found to coincide with the known results for the complex Ginibre ensemble (GinUE), in the region of eigenvalue depletion close to the real axis the asymptotic for the GinOE is considerably different. We also study numerically the distribution of diagonal overlaps and conjecture that it is the same in the bulk and at the edge of both the GinOE and GinUE, but essentially different in the depletion region of the GinOE.
\end{abstract}


\section{Introduction}\label{sec:intro}

Random matrices of finite size $N\times N$ are often categorized in terms of their global symmetries, influencing  statistical properties of their eigenvalues and eigenvectors.  The two major categories to distinguish are matrices which are self-adjoint (Hermitian) from their non-self-adjoint (non-Hermitian) counterparts. The former are defined as satisfying the condition  $H=H^\dagger \equiv \bar{H}^T$, with $^\dagger$ standing for Hermitian conjugation,  $^T$ denoting the transpose of the matrix and  the bar standing for the  complex conjugation of its entries.  All Hermitian matrices are by definition {\it normal}, with vanishing commutator $[H,H^\dagger]=0$, which ensures that $H$ is diagonalizable by a unitary transformation. The Hermiticity of $H$ ensures all ensuing eigenvalues $\lambda_n$ to be necessarily real.

On the other hand, \emph{non-Hermitian} matrices, which are denoted by $X$,  may or may not be normal, in the latter case $XX^\dagger \neq X^\dagger X$. Such matrices generically have the majority of the eigenvalues $z_n$ ($n=1,\ldots,N$) located in the full complex plane $\mathbb{C}$, although some of them may still remain on the real line $\mathbb{R}$. The non-normality of such matrices is well known to lead to important numerical issues, when the matrix size $N$ grows. For example, keeping a fixed precision of calculations might not be sufficient, as some eigenvalues can be ‘ill-conditioned’, i.e. they are much more sensitive to a perturbation of the matrix entries, in contrast to eigenvalues of normal matrices, which are relatively robust.

For a random matrix $X$, we can safely assume that all its $N$ eigenvalues have multiplicity one,  ensuring that $X$ is still diagonalizable by  a transformation to its eigenbasis: $X=S\Lambda S^{-1}$, with $\Lambda = \diag(z_1,\ldots,z_N)$. However for non-normal matrices the matrix $S$ ceases to be unitary: $S^\dagger \neq S^{-1}$. The associated eigenvalue problem then splits and one needs to distinguish between \emph{left} (row) eigenvectors, satisfying $\mathbf{x}_{L_n}^\dagger X = z_n \mathbf{x}_{L_n}^\dagger$ and \emph{right} (column) eigenvectors, which instead satisfy $X \mathbf{x}_{R_n} = z_n \mathbf{x}_{R_n}$. The right eigenvectors are the columns of the matrix $S$, whereas the left eigenvectors are the rows of $S^{-1}$, and generically $\mathbf{x}_{R_n} \neq \mathbf{x}_{L_n}$. In this representation, the sets of left and right eigenvectors automatically satisfy the bi-orthonormality condition $\mathbf{x}_{L_n}^\dagger \mathbf{x}_{R_m} = \delta_{nm}$, for $n,m=1,\ldots,N$, but non-unitarity of $S$ implies that $\mathbf{x}_{L_n}^\dagger \mathbf{x}_{L_m} \neq \delta_{nm}$ and similarly $\mathbf{x}_{R_n}^\dagger \mathbf{x}_{R_m} \neq \delta_{nm}$. The first object that allows qualitative insights into the non-orthogonality of the eigenvectors, is the \textit{overlap matrix}, with entries $\mathcal{O}_{nm} = \left( \mathbf{x}_{L_n}^\dagger \mathbf{x}_{L_m} \right) \left( \mathbf{x}_{R_m}^\dagger \mathbf{x}_{R_n} \right)$. In physical literature the diagonal entries, $\mathcal{O}_{nn}$, are generally referred to as diagonal overlaps, or \textit{self-overlaps}. In numerical analysis literature their square roots are sometimes referred to as the \textit{eigenvalue condition numbers}, as they characterize the sensitivity of the eigenvalue $z_n$ to a perturbation of the entries of $X$, see e.g. \cite{TE,TTRD}. To this end, consider a matrix $X^\prime = X + \varepsilon P$, the second term reflecting an error made with respect to the entries of $X$, which are stored with finite machine precision. $P$ here can be chosen to have its $2$-norm fixed to $\vert \vert P\vert \vert_{2}=1$, such that the strength of such a perturbation is measured via the real parameter $\varepsilon$. Then, the change in eigenvalue $z_n$ of $X$ incurred by a small $\varepsilon$ can be characterized via 
\be\label{perturbIntro}
    \frac{dz_n(\varepsilon)}{d\varepsilon} \bigg\vert_{\varepsilon=0}: = \dot{z}_n(0)=\mathbf{x}_{L_n}^\dagger P \mathbf{x}_{R_n} \ ,
\ee
which can be derived using the bi-orthonormality of the corresponding left- and right-eigenvectors. By using the Cauchy-Schwartz inequality, the magnitude of the shift can be seen to satisfy the following bound:
\be\label{perturbIntro2}
    \vert \dot{z}_n(0) \vert =  \vert \mathbf{x}_{L_n}^\dagger P \mathbf{x}_{R_n} \vert \leq \vert \mathbf{x}_{L_n} \vert \ \vert \vert P\vert \vert_{2} \ \vert \mathbf{x}_{R_n} \vert = \sqrt{(\mathbf{x}_{L_n}^\dagger \mathbf{x}_{L_n}) (\mathbf{x}_{R_n}^\dagger \mathbf{x}_{R_n})} = \mathcal{O}_{nn}^{1/2} \ ,
\ee
where $\mathcal{O}_{nn}$ are the diagonal entries of the overlap matrix.  We indeed see that the magnitude of the perturbation is essentially controlled by $\mathcal{O}_{nn}^{1/2}$, which assumes the minimal value  only if $X$ is normal and all $O_{nn}=1$. In the realm of non-normal matrices one can find numerous examples where $\mathcal{O}_{nn}\gg 1$, indicating that the associated eigenvalues are extremely sensitive to perturbations, see e.g. \cite{TE,TTRD}. In fact, random matrices can be used to provide a regularization of eigenvalue condition numbers of highly non-normal, non-random matrices, see for example \cite{Banks_perturb,Cipolloni23c,Jain_perturb} for more information.

For $X$ taken to be a random matrix from a specified probability measure, the statistics of the entries of the overlap matrix $\mathcal{O}_{mn}$ become an important object of study. The simplest nontrivial choice is to assume that all entries $X_{ij}$ are mean-zero, independent, identically distributed (i.i.d.) Gaussian numbers, which can be real, complex, or quaternion. This defines the three classical Ginibre ensembles \cite{Ginibre}. Understanding of statistics of the overlap matrix in this setting has been  influenced heavily by the seminal papers of Chalker and Mehlig \cite{CM,MC} who treated the \emph{complex Ginibre} ensemble. We denote the latter ensemble as GinUE in the rest of the paper, to distinguish it from its real (GinOE) and quaternion (GinSE) counterparts, the nomenclature relating to the classical Dyson Hermitian ensembles - GUE, GOE and GSE. In particular, Chalker and Mehlig addressed the statistics of the overlap matrix $\mathcal{O}_{mn}$ via considering the following single-point and two-point correlation functions
\be\label{ChalkerMehligOverlap}
    \mathcal{O}(z) \equiv \bigg\langle \frac{1}{N}\sum_{n=1}^N \mathcal{O}_{nn} \ \delta(z-z_n) \bigg\rangle \hspace{0.75cm}  \text{and} \hspace{0.75cm} \mathcal{O}(z_1,z_2) \equiv \bigg\langle \frac{1}{N^2}\sum_{\substack{n,m=1 \\ n\neq m}}^N \mathcal{O}_{nm} \ \delta(z_1 -z_n) \  \delta(z_2-z_m) \bigg\rangle\ ,
\ee
where the angular brackets stand for the expectation value with respect to the probability measure associated with the ensemble in question,  GinUE for the particular case studied by Chalker and Mehlig. Here $\delta (z-z_n)$ is the Dirac delta mass at the eigenvalue $z_n$, so that the empirical density of eigenvalues in the complex plane $z$ reads $\rho^{(\text{emp})}_N(z) = \frac{1}{N} \sum_{n=1}^N \delta(z - z_n)$. It is evident that $\mathcal{O}(z)$ describes the conditional expectations of $\mathcal{O}_{nn}$ and we can define the mean conditional self-overlap as 
\be\label{conditional}
    \mathbb{E}\left( z  \right) \equiv \mathbb{E}\left(\mathcal{O}_{nn} \ \vert \ z = z_n \right) = \frac{\mathcal{O}(z)}{ \rho\, (z) } \ \text{,}
\ee
where $\rho\, (z)$ is the mean spectral density, defined via
\begin{equation}\label{meanden}
    \rho\, (z) \equiv \langle \ \rho_N^{(\text{emp})}(z) \ \rangle = \bigg\langle \frac{1}{N}\sum_{n=1}^N  \delta(z-z_n) \bigg\rangle \ . 
\end{equation}
Recall that, as $N\to \infty$, the asymptotic mean eigenvalue density is nonvanishing and uniform only inside the unit circle in the complex plane:  $ \rho (z)  \approx 1/\pi $ for $\vert z \vert^2 <  1 $ and zero otherwise \cite{Girko}. Chalker and Mehlig were able to extract the leading asymptotic behaviour of $\mathcal{O}(z)$ and $\mathcal{O}(z_1,z_2)$ as $N$ tends to $\infty$. Choosing the variance of the entries of $X$ to be $1/N$, they found that $\mathcal{O}(z) \approx N \left(1-\vert z \vert^2 \right)/\pi $ inside the unit disk $\vert z \vert^2 <  1 $ and zero otherwise.  Consequently, one should expect $\mathcal{O}_{nn} \sim N$ for eigenvalues inside the disk, which is thus parametrically larger than the value for normal matrices, where $\mathcal{O}_{nn} = 1$.  Note that only the rescaled correlator $ \widetilde{\mathcal{O}}(z) = \mathcal{O}(z)/N $ is well-defined and finite in the limit $N\to \infty$.
 
The Chalker-Mehlig (CM) correlators and other aspects pertaining to eigenvector non-orthogonality in random matrix ensembles have attracted growing interest in the theoretical physics community over the past decades, see e.g. \cite{JNNPZ,SFPB,MS,FM,GS,FSav2,GKLMRS, BGNTW,BGNTW2,BSV,BNST,NT,AFK,FyoOsm22}. The main interest being in extending CM correlators from the GinUE to more general classes of random matrix models.  Apart from the important issues of eigenvalue stability, the eigenvector non-orthogonality is known to play an important role in describing transient behaviour in complex systems with classical dynamics and related questions \cite{GHS,Grela,EKR,MBO,TNV, GOCNT,FGNNW}, as well as in producing intriguing features in their quantum counterparts \cite{Cipolloni23a,Cipolloni23b,GKR23}. 
Another strong motivation comes from the field of quantum chaotic wave scattering, where non-Hermitian random matrix ensembles, different from Ginibre, play a prominent role, see e.g. \cite{FSom1,FSom2,Rotter, FSav1} for some background information. The overlap matrix $\mathcal{O}_{mn}$ appears naturally in many scattering observables, for example in the derivation of decay laws \cite{SS}, in the ‘Petermann factors’ describing excess noise in open laser resonators \cite{SFPB,PSB}, in issues of increased sensitivity of resonance widths to small perturbations \cite{FSav2,GKLMRS} and in the shape of reflected power profiles \cite{FyoOsm22}. In that context both $\mathcal{O}(z)$ and $\mathcal{O}(z_1,z_2)$ have been studied theoretically \cite{SFPB,MS,FM,FSav2,FyoOsm22} and experimentally \cite{GKLMRS,CG,CG2}. 

In a not unrelated development, mathematically rigorous studies of CM correlators became a field of considerable activity more recently, see e.g. \cite{Dubach21,Dubach23,WS}.  The usage of free probability techniques \cite{JNNPZ} to compute the self-overlap, Eq. \eqref{ChalkerMehligOverlap}, allowed extensions to invariant ensembles with non-Gaussian weights \cite{BNST,NT}. Since $\mathcal{O}(z)$ is known in the GinUE at finite matrix size \cite{CM,MC} it also became possible to compute it for products of small Ginibre matrices \cite{BSV}. It was shown that one- and two-point functions of eigenvector overlaps, conditioned on an arbitrary number of eigenvalues in the GinUE, lead to determinantal structures \cite{ATTZ}. Deep insights into Dysonian dynamics related to eigenvalues have been made, from different angles, in   \cite{BGNTW,GW,BD,Yabuoku}.
Lower \cite{Cipolloni23c,Cipolloni22} and upper \cite{Banks_perturb,Jain_perturb,ErdJi} bounds on diagonal eigenvector overlaps have been provided for a fixed non-Hermitian matrix perturbed by random matrices with i.i.d. entries, even beyond Gaussian case. Some other properties of eigenvectors of non-normal random matrices have also been studied rigorously, see  \cite{BZ,CR,Esaki,Movassagh}.

Finally, let us mention that in addition to analysing the correlators in Eq. \eqref{ChalkerMehligOverlap}, Chalker and Mehlig put forward a conjecture for the far tail behaviour of the distribution of the random variable $\mathcal{O}_{nn}$ for the GinUE as $N \rightarrow \infty$. Based on both numerical evidence and simple eigenvalue repulsion arguments, illustrated by a solvable $2\times 2$ GinUE matrix, they predicted that for large overlaps the probability density of $\mathcal{O}_{nn}$ must exhibit a power law tail proportional to $1/\mathcal{O}_{nn}^{3}$.  Such tail would make all the positive integer moments beyond $\mathcal{O}(z)$ divergent. This conjecture has been rigorously proved by two different approaches in Bourgade and Dubach \cite{BD} and Fyodorov \cite{FyodorovCMP}, in fact recovering the full form of the probability density beyond the tail region. While the work of Bourgade and Dubach proceeded on studying the off-diagonal correlator $\mathcal{O}(z_1,z_2)$ for the GinUE, the paper by Fyodorov revealed that for real eigenvalues of the GinOE the diagonal overlaps $\mathcal{O}_{nn}$ are distributed with an even heavier probability density tail, decaying as $1/\mathcal{O}_{nn}^{2}$. This implies that for real Ginibre matrices the mean self-overlap $\mathcal{O}(z)$ is divergent on the real line. Understanding both the distribution  of $\mathcal{O}_{nn}$ and the CM mean $\mathcal{O}(z)$ for complex eigenvalues of the GinOE remained however an outstanding problem. \\

\noindent
In the present paper we make the first step towards addressing the above issues, and present the results for the mean diagonal CM correlator $\mathcal{O}(z)$ for the GinOE in the complex plane, first at finite $N$ and then in various scaling regimes as $N\gg 1$. We will also systematically  compare our findings with both results for complex eigenvalues in the GinUE and for real eigenvalues in the GinOE.

The rest of the paper is organized as follows. We present our main findings in Section \ref{sec:MainRes}, in particular, the mean self-overlap at finite matrix size $N$ is given in Theorem \ref{thm:MainRes}. Asymptotic results for the mean self-overlap, as $N\to \infty$, are given in the bulk, Corollary \ref{cor:GinOEbulkStrong}, at the spectral edge, Corollary \ref{cor:GinOEedgeStrong}, and in an eigenvalue depleted region of the droplet in Corollary \ref{cor:GinOEdepletionStrong}. We then compare with similar results for the GinUE, and use numerical simulations, both to corroborate analytical results and provide insights into yet analytically unavailable distributions of $O_{nn}$ for the GinOE. Finally, the corresponding proofs of our findings are presented in Section \ref{sec:MainResDerivation} for finite $N$ - Theorem \ref{thm:MainRes} - and in Section \ref{sec:AsymptoticAnalysis} for the asymptotic results - Corollary \ref{cor:GinOEbulkStrong}, \ref{cor:GinOEedgeStrong} and \ref{cor:GinOEdepletionStrong}.


\section{Statement and Discussion of Main Results}\label{sec:MainRes}

\noindent
In order to state our main results, we start by introducing the necessary notation and statements about random real Ginibre matrices in Section \ref{subsec:DefRem}. We then present our main results for the mean self-overlap in Section \ref{subsec:MainRes} for both finite matrix size $N$ and in several large $N$ regimes of the complex plane. We follow up with a discussion about connections to previously known results, comparisons to numerical simulations and open problems.


\subsection{Remarks on real and complex Ginibre ensembles}
\label{subsec:DefRem}

\begin{defn}\label{def:GinOE}
    Let $G = \left( G_{ij} \right)_{i,j=1}^N$ be an $N \times N$ matrix containing i.i.d. real Gaussian entries with mean zero and unit variance, such that the off-diagonal entries $G_{ij}$ and $G_{ji}$ are uncorrelated. 
    The joint probability density function (JPDF) of matrices $G$ is defined with respect to the flat Lebesgue measure, $dG = \prod_{i,j=1}^{N} dG_{ij}$, via 
    \be\label{GinDistrib}
        P_{\text{GinOE}}(G) \ dG = \frac{1}{C_{N}} \ \exp \left[ -\frac{1}{2} \Tr \left( GG^T \right)  \right] \ dG, \quad  C_{N} = (2\pi)^{N^2/2} \ \text{.}
    \ee
     The corresponding ensemble is called the \emph{real Ginibre} ensemble, and is denoted as GinOE. Its counterpart for matrices with complex i.i.d. entries is the \emph{complex Ginibre} ensemble (GinUE). For recent reviews of available results on the GinUE and GinOE see \cite{BF1,BF2}. The remarks below provide a discussion of main facts on both ensembles which will be of direct relevance for the present paper. 
    
\end{defn}

\begin{rem}\label{rem:GinOEDef}
     Generically, the spectrum of GinOE matrices consists of real eigenvalues $\lambda\in \mathbb{R}$ and complex eigenvalues $z=x+iy, \, y\ne 0$. Complex eigenvalues of real matrices, $G$, always appear in conjugated pairs, i.e. $\bar{z}=x-iy$ is an eigenvalue of $G$ iff $z$ is an eigenvalue of $G$. Correspondingly, the mean spectral density, Eq. \eqref{meanden}, necessarily has the form $\rho\,(z) = \rho^{(c)}(z)+\delta(y) \rho^{(r)}(x)$, where $\rho^{(c/r)}(z)$ describes the mean density of complex/real eigenvalues, respectively. On the other hand, with probability one, the eigenvalues of GinUE matrices are complex without conjugate counterparts, due to the entries of GinUE matrices being complex.
    
     In Figure \ref{fig:spectra_Real_Gin}, we show examples of the spectra of real Ginibre matrices for three different values of $N$.  As is well-known, the majority of complex eigenvalues lie inside a circle of radius $\sqrt{N}$, with a subleading fraction (of the order of  $1/\sqrt{N}$ for $N\gg 1)$ of purely real eigenvalues. As $N$ tends to $\infty$ the support of the spectrum approaches a uniform disc.  

    \begin{figure}[h]
        \centering
        \includegraphics[scale = 0.235]{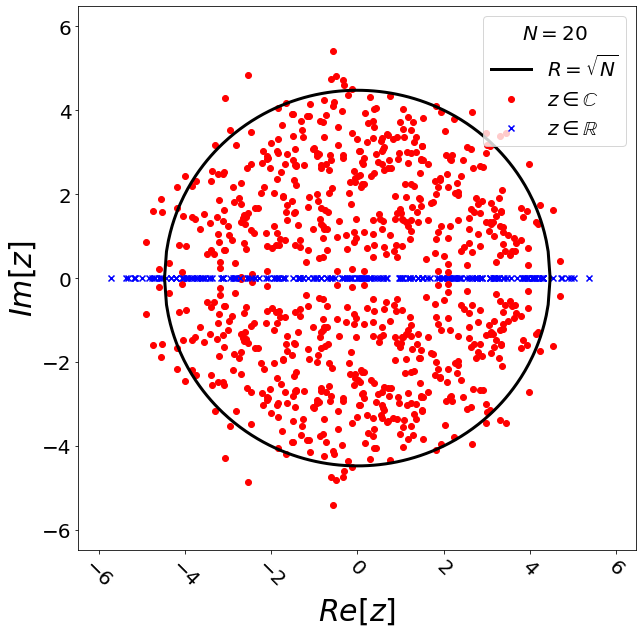}
        \hspace{0.1cm}
        \includegraphics[scale = 0.235]{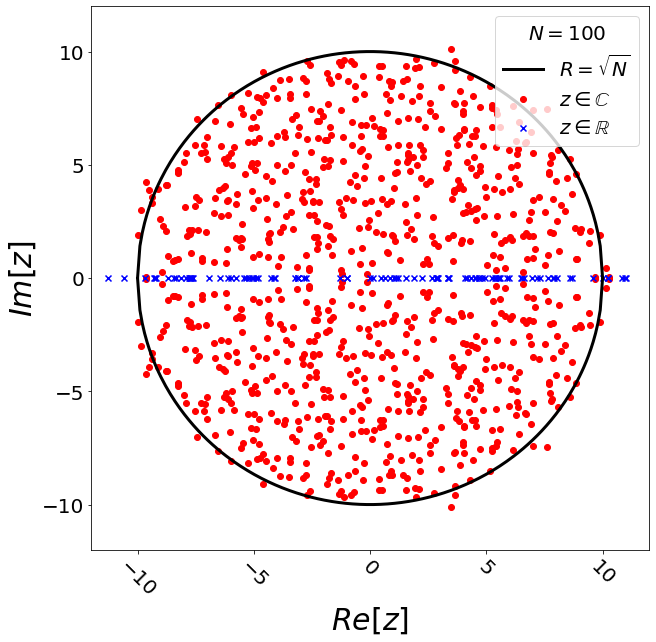}
        \hspace{0.1cm}
        \includegraphics[scale = 0.235]{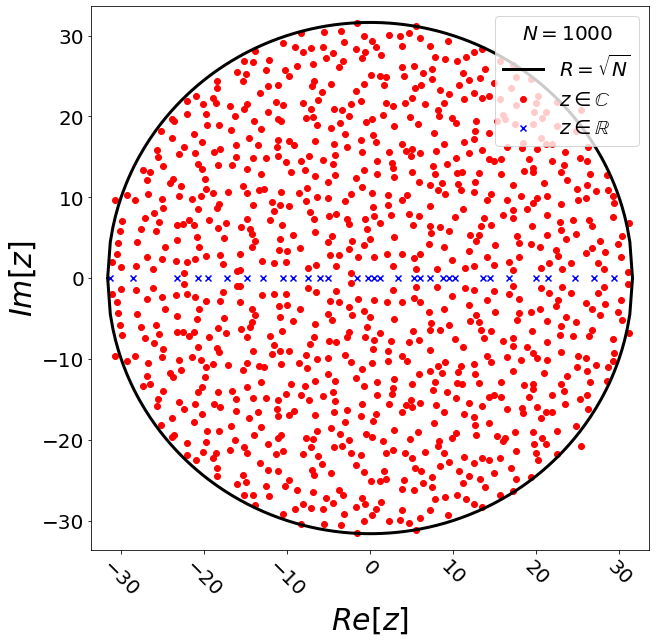}
        
        \caption{\small Spectra of real Ginibre matrices for three different values of $N$, $N=20$ (left), $N=100$ (centre) and $N=1000$ (right). Each plot contains 1,000 samples of eigenvalues and has a solid black line depicting a circle of radius $\sqrt{N}$. Complex eigenvalues are shown in red and real eigenvalues are shown in blue.}
        \label{fig:spectra_Real_Gin}
    \end{figure}
    \end{rem}
     
   \begin{rem} 
   Results for the mean densities (both real and complex) are well known in the GinOE and GinUE at finite $N$ \cite{Ginibre,EKS,Edelman,KS}. The mean density of complex eigenvalues is of particular interest in this work and is given for finite $N$ in the GinOE by
    \be\label{GinOEdensityComplex}
    \begin{split}
        \rho_N^{(\text{GinOE,c})}(z) &= \sqrt{\frac{2}{\pi}} \ \vert y \vert \ e^{2y^2} \ \text{erfc}\left(\sqrt{2}\ \vert y \vert \right) \ \frac{\Gamma\left(N-1, \vert z \vert^2\right)}{\Gamma\left(N-1 \right)} \ ,
    \end{split}
    \ee 
    see e.g. \cite[Eq. (2.46)]{BF2} and \cite[Theorem 6.2]{Edelman}.  In the above, $\Gamma(N)$  denotes the standard Euler $\Gamma$-function and $\Gamma(N,|z|^2)$  denotes the incomplete (upper) $\Gamma$-function defined as 
    \be\label{Eq:incompleteGamma}
        \Gamma\left( N , a \right) = \Gamma\left(N\right) \ e^{-a} \ \sum_{k=0}^{N-1} \frac{a^k}{k!} =  \int_{a}^{\infty} du \ u^{N-1} \ e^{-u} \ .
    \ee
    We have also made use of the complementary error function ${\erfc(x) = 1 - \erf(x)}$, where $\erf(x) = \frac{2}{\sqrt{\pi}} \int_0^x e^{-t^2} dt$.
  
    \noindent
    In the GinUE, several equivalent representations for the mean density can be found in e.g. \cite{Ginibre}, \cite[Eq. (20)]{WS}, \cite[Proposition 2.2]{BF1} and \cite[Eq. (18.2.11)]{KS}, which we prefer to write as
    \be
        \rho^{(\text{GinUE},c)}_N(z) = \frac{1}{\pi} \ \frac{\Gamma\left( N, \vert z \vert^2 \right)}{\Gamma\left( N \right)} = \frac{1}{\pi } e^{- \vert z \vert^2} \sum_{n=0}^{N-1}  \frac{ \vert z \vert^{2n}}{n!} \ .
        \label{eq:GinUE_density}
    \ee
  The mean density of \emph{real} eigenvalues in the GinOE is also known at finite $N$, see e.g. \cite{EKS,FN}, but is not needed for our purposes.

    \noindent
    
\end{rem}

\begin{rem} \label{rem:asymp_dens}

    It is worth discussing the large $N$ asymptotic behaviour of the mean eigenvalue density in more detail. A comparison between the densities of complex eigenvalues in the GinOE and GinUE at large $N$ is shown in Figure \ref{fig:Heatmaps} using heatmaps. The heatmap associated with the GinUE features two distinct scaling regimes: a spectral bulk inside the disk and edge along the disk circumference. Similar regimes are also seen in the plot for the GinOE. 

    \begin{figure}[h!]
        \centering
        \includegraphics[height = 6.5cm]{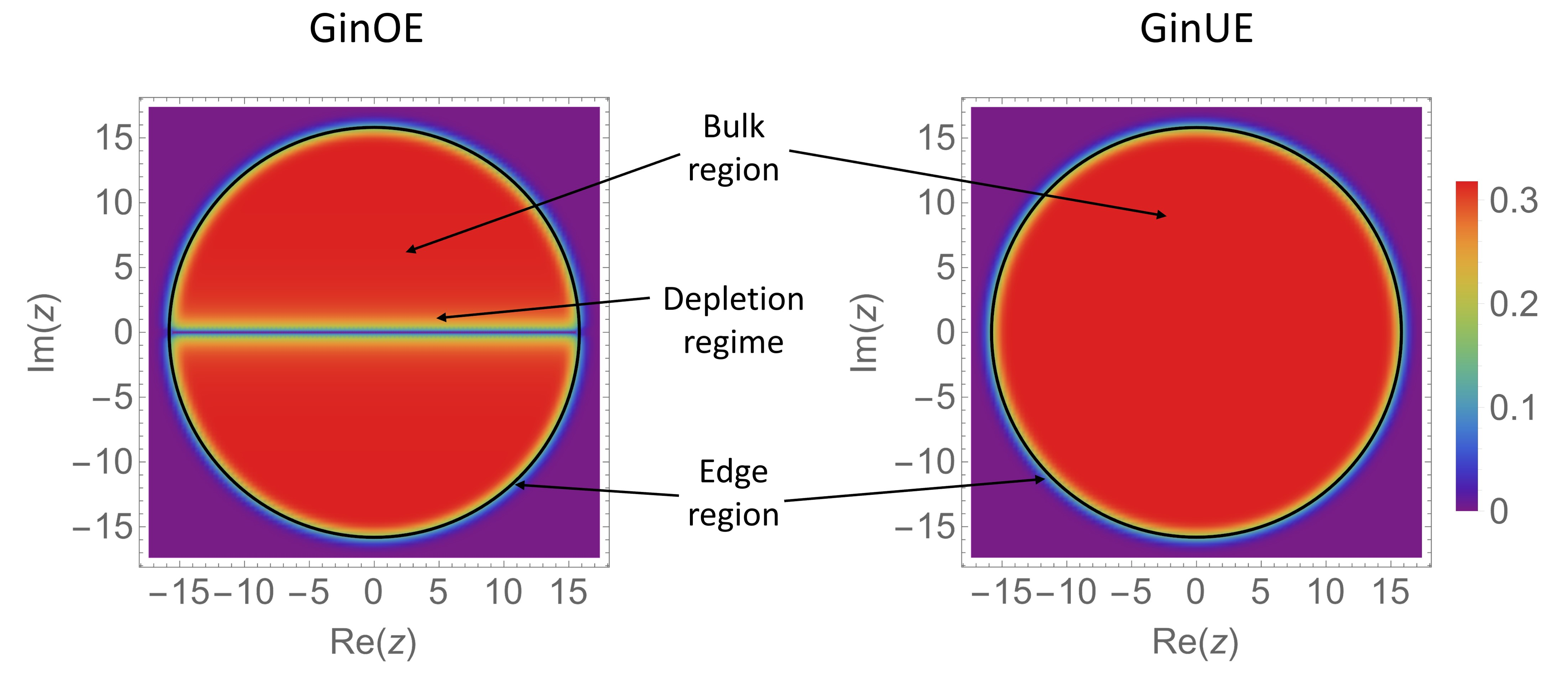}
        
        \caption{\small Large $N$ density of complex eigenvalues in the GinOE and GinUE with relevant scaling regions labelled. Left: GinOE. Right: GinUE. In both diagrams the solid black circle within the heatmap has radius $\sqrt{N}$ and the density is depicted on a rainbow scale. This diagram explains schematically what is meant by the bulk, edge and depletion regimes of large real and complex Ginibre matrices. Each heatmap is plotted using the analytic expression for the density of complex eigenvalues in the associated ensembles when $N=250$.}
        \label{fig:Heatmaps}
    \end{figure}

    \noindent
    Firstly, in both ensembles it is easy to show that the limiting mean density of complex eigenvalues, after rescaling  $z=\sqrt{N}w$,  
    converges to the uniform density:
    \be
        \rho_{\text{bulk}}^{(\GinUE, c)}(w) =\rho_{\text{bulk}}^{(\GinOE,c)}(w) =  \frac{1}{\pi} \ ,
    \ee
    in the generic bulk of the unit disk, where $\vert w \vert<1$. Similarly, at the edge of the disk, given by  $|z| = \sqrt{N} + \eta$,  the density for both ensembles can be shown to be given by
    \be
        \rho_{\text{edge}}^{(\GinUE, c)}(\eta) =\rho_{\text{edge}}^{(\GinOE,c)}(\eta) = \frac{1}{2 \pi} \erfc\left(\sqrt{2} \eta \right) \ .
    \ee
    In addition, the heatmap in the GinOE case shows yet another scaling regime close to the real axis, essentially for heights of the fraction $\sim 1/\sqrt{N}$ of the disk radius. In that region the density of eigenvalues is reduced when compared to the spectral bulk value. The origin of such a depletion is clearly seen analytically from the presence of the factor $\vert y\vert$ in  Eq. \eqref{GinOEdensityComplex}. Henceforth the scaling associated with this region of the complex plane will be referred to as the \emph{depletion} regime. Inside the bulk depletion regime, i.e. $z=\sqrt{N}x+iy$, with $|x|<1$ and $y \sim O(1)$, the mean GinOE density in the limit $N\to \infty$ converges to
    \be
        \rho_{\text{depletion}}^{(\GinOE,c)}(y) = \sqrt{\frac{2}{\pi}} \  |y| \ e^{2y^2} \erfc\left( \sqrt{2} \ |y| \right) \ ,
    \ee
    see e.g. 
    \cite[Eq. (2.42)]{BF2}, cf. Eq. \eqref{GinOEdensityComplex}.

\end{rem}    

\noindent
After the digressions on the mean densities, let us come back to our main objects of interest, the left and right eigenvectors.  Given the complex eigenvalue $z_n$ we denote the associated left-eigenvector by $\mathbf{x}_{L_n}^\dagger$ and the right-eigenvector by $\mathbf{x}_{R_n}$. The overlap matrix $\mathcal{O}_{nm} = \left( \mathbf{x}_{L_n}^\dagger \mathbf{x}_{L_m} \right) \left( \mathbf{x}_{R_m}^\dagger \mathbf{x}_{R_n} \right)$ of left- and right-eigenvectors is used to define the mean self-overlap $\mathcal{O}(z)$ as in Eq. \eqref{ChalkerMehligOverlap} and its conditional companion $\mathbb{E}(z)$ as in Eq. \eqref{conditional}. Note that $\mathbb{E}\left( z  \right)$ is particularly useful when comparing theoretical predictions with numerical data.\\

\noindent
Results for the mean self-overlap have been established in the GinUE for the entire complex plane, see \cite{CM,MC,WS,BD,FyodorovCMP} for several equivalent forms.  For our purposes we present it as
\be\label{Eq:Overlap_finiteN_GinUE}
    \mathcal{O}^{(\text{GinUE,c})}_N(z) = \frac{1}{\pi} \left[ \frac{\Gamma(N, |z|^2)}{(N-1)!}(N - |z|^2) + \frac{|z|^{2N}}{(N-1)!} e^{-|z|^2} \right] \ .
\ee
For convenience of the reader we give a derivation of Eq. \eqref {Eq:Overlap_finiteN_GinUE} in the Appendix \ref{AppA}, using the results in \cite{FyodorovCMP} as a starting point. 
One can further easily find the bulk asymptotics of this expression, which gives back the original Chalker-Mehlig result,
as well as the corresponding edge asymptotics for the scaling $|z|=\left( \sqrt{N} + \eta  \right)$ given by  
\be\label{edgeasyGinUE}
   \lim_{N\to \infty} \frac{1}{\sqrt{N}}\mathcal{O}_{\textup{edge}}^{\textup{(GinUE,c)}}(z) = \frac{1}{\pi} \left( \frac{1}{\sqrt{2\pi}} \ e^{-2\eta^2}  - \eta \ \textup{erfc}\left( \sqrt{2} \ \eta \right) \right)  \ .
\ee
In the GinOE, so far, results were limited to real eigenvalues only, i.e. when $z=x$, see \cite{FyodorovCMP, FT}.
In the next subsection, we provide results for the mean self-overlap $\mathcal{O}(z)$ of eigenvectors associated with complex eigenvalues in the GinOE at finite $N$ and in several large $N$ scaling regions, dictated by the corresponding scalings of the mean eigenvalue density discussed above.


\subsection{Statement of Main Results for GinOE and comparison to GinUE}
\label{subsec:MainRes}

 Relegating the proofs and technical details to Section \ref{sec:MainResDerivation}, we present our main findings below. The following theorem gives the mean self-overlap $\mathcal{O}(z)$ of eigenvectors associated with a complex eigenvalue $z$ for the real Ginibre ensemble at finite matrix size $N$.

\begin{thm}\label{thm:MainRes}
    Let $G$ be an $N \times N$ random matrix drawn from the GinOE, distributed according to Eq. \eqref{GinDistrib} in Definition \ref{def:GinOE}. The mean self-overlap, Eq. \eqref{ChalkerMehligOverlap}, associated with a complex eigenvalue $z$ at finite matrix size $N$ is given by
    \be\label{Eq:MainResGinOE}
    \begin{split}
        \mathcal{O}^{(\GinOE,c)}_{N}(z) &= \bigg\langle \frac{1}{N}\sum_{n=1}^N \mathcal{O}_{nn} \ \delta(z-z_n) \bigg\rangle_{\textup{GinOE},N}  = \frac{1}{\pi} \ \left( 1 + \sqrt{\frac{\pi}{2}} \ \exp \left[ 2y^2 \right] \ \frac{1}{2 \vert y \vert} \ \textup{erfc}\left(\sqrt{2} \ \vert y \vert \right) \right) \\
        &\times \bigg[  \ \frac{\Gamma\left(N-1,\vert z \vert^2 \right)}{(N-2)!}  \bigg( N - 1 - \vert z \vert^2 \bigg) +  \frac{\vert z\vert^{2(N-1)}}{(N-2)!}  \ e^{-\vert z \vert^2} \bigg] \ .
    \end{split}
    \ee
\end{thm}

\noindent
Comparing Eq. \eqref {Eq:MainResGinOE} with its GinUE counterpart Eq. \eqref {Eq:Overlap_finiteN_GinUE}, we see 
that they almost share the term inside the square brackets (with the change $N \to N-1$). The result in the GinUE is fully rotationally symmetric as there is no $y-$dependent factor present. In the GinOE however, the self-overlap is dependent on the distance to the real axis and so is not rotationally symmetric.

\noindent
In Figure \ref{fig:finite_N}, we compare the  mean conditional self-overlap, Eq. \eqref{conditional}, for eigenvectors associated with purely imaginary eigenvalues, $z = iy$, in the GinOE and GinUE at finite $N$, as predicted by our theory and as seen in numerical simulations. As is evident, in the region close to the real line, the mean conditional self-overlap is much larger in the GinOE than the GinUE, implying that the GinOE has a higher degree of non-normality.
    
\begin{figure}[H]
    \centering
    \includegraphics[scale = 0.215]{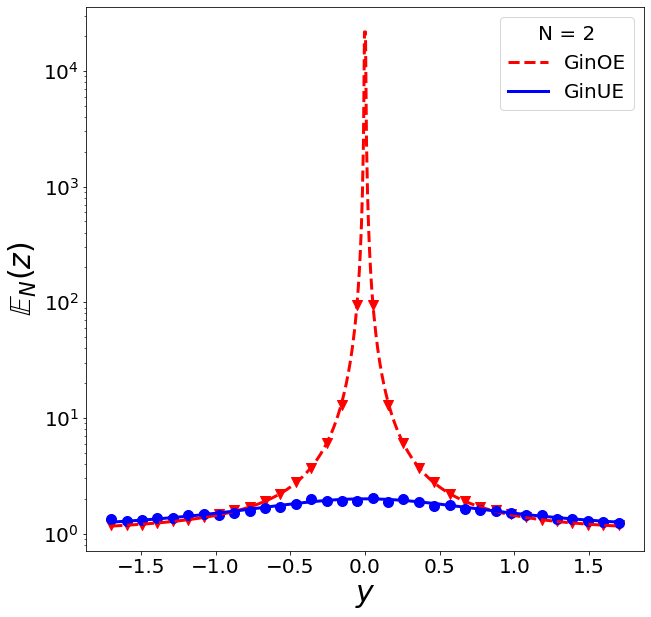}
    \hspace{0.1cm}
    \includegraphics[scale = 0.215]{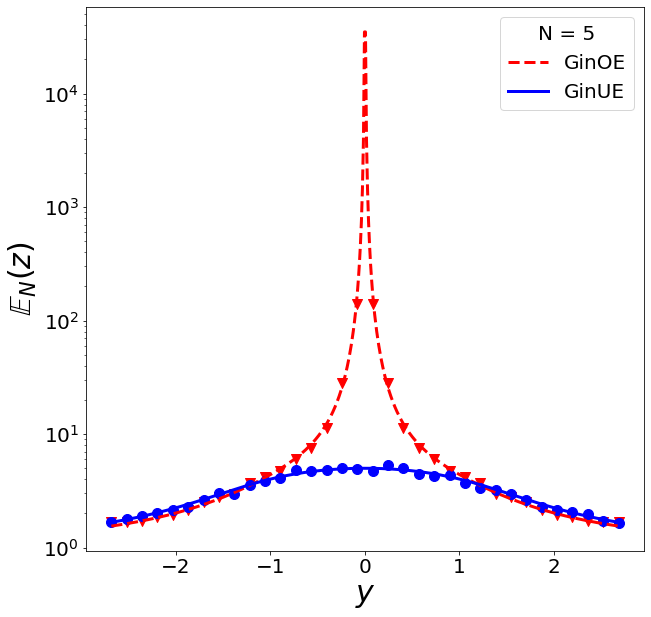}
    \hspace{0.1cm}
    \includegraphics[scale = 0.215]{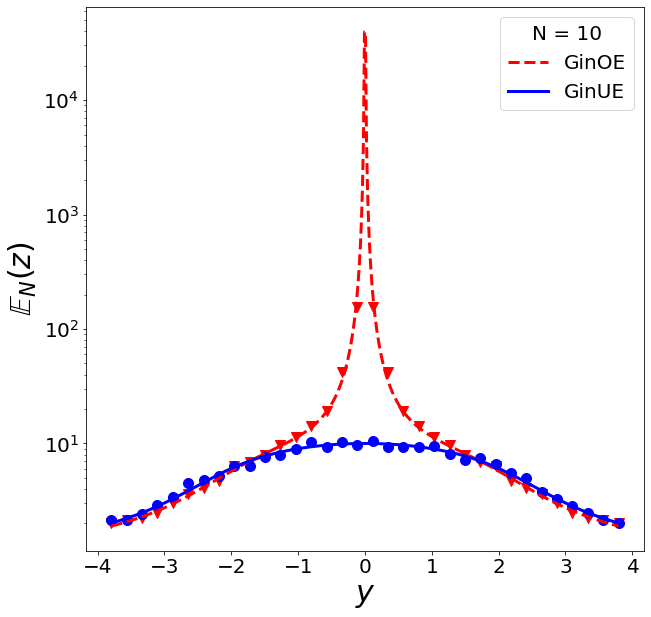}
    
    \caption{\small Mean conditional self-overlap $\mathbb{E}_N(z)$ associated with purely imaginary eigenvalues in the GinOE and GinUE, for three different values of $N$. Theoretical predictions are red dashed lines (GinOE) and blue solid lines (GinUE). Triangular (circular) markers represent numerical results for the GinOE (GinUE). The mean self-overlap is measured by averaging self-overlaps of the closest $O(10^3)$ eigenvalues with respect to the chosen value of $y$.  Numerically averaged values are taken from a data set containing $O(10^8)$ samples of the self-overlap. 
    }\label{fig:finite_N}
\end{figure}

\noindent
Next, we provide results for the large $N$ asymptotic behaviour of the mean self-overlap in the bulk, at the spectral edge and in the depletion regime, as defined in the previous section.

\begin{cor}\label{cor:GinOEbulkStrong}
    For a complex eigenvalue $z = \sqrt{N}w$, where $w=x+iy$  with $\vert w \vert < 1$ while keeping  $|y| \gg N^{-1/2}$ as $N\to \infty$, the limiting scaled mean self-overlap in the bulk is given by
    \be\label{Eq:GinOEbulkStrongRes}
        \mathcal{O}_{\textup{bulk}}^{\textup{(GinOE,c)}}(w) \equiv \lim_{N\rightarrow \infty} \frac{1}{N} \ \mathcal{O}_N^{\textup{(GinOE,c)}}\left( \sqrt{N}w \right) 
        = \frac{1}{\pi} \left( 1 - \vert w \vert^2 \right) \ \Theta\left[ 1 - \vert w \vert^2 \right] \ ,
    \ee
    where $\Theta\left[ a \right]$ is the Heaviside function, which is equal to one, when $a>0$ and zero otherwise. 
\end{cor}

\noindent
This is nothing else but the original Chalker-Mehlig asymptotic for the mean self-overlap in the spectral bulk of the GinUE. This relation has been tested in both the GinOE and GinUE using numerical simulation, with the results shown in Figure \ref{fig:bulk_Gin}. In the case of the GinOE, one can see that deep inside the bulk region the formula Eq. \eqref{Eq:GinOEbulkStrongRes} is accurate for $N \geq O(10^2)$, indicating that the theory is accurate as $N \to \infty$.

\begin{figure}[H]
    \centering
    \includegraphics[scale = 0.215]{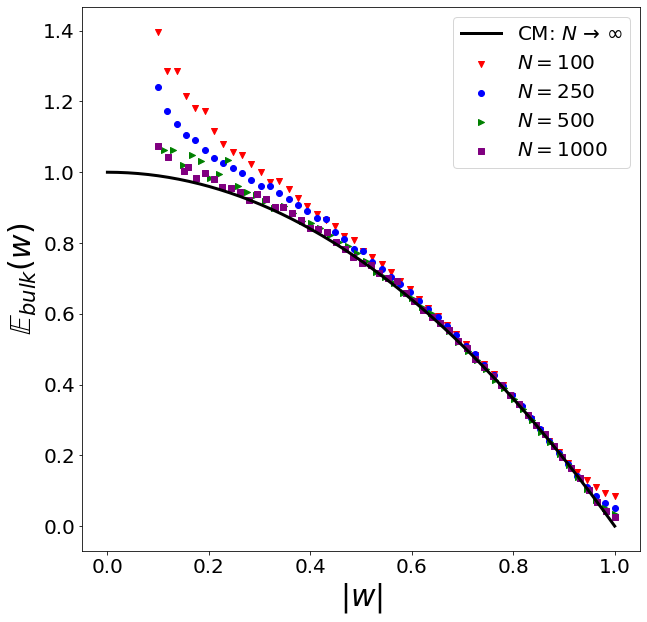}
    \hspace{2cm}
    \includegraphics[scale = 0.215]{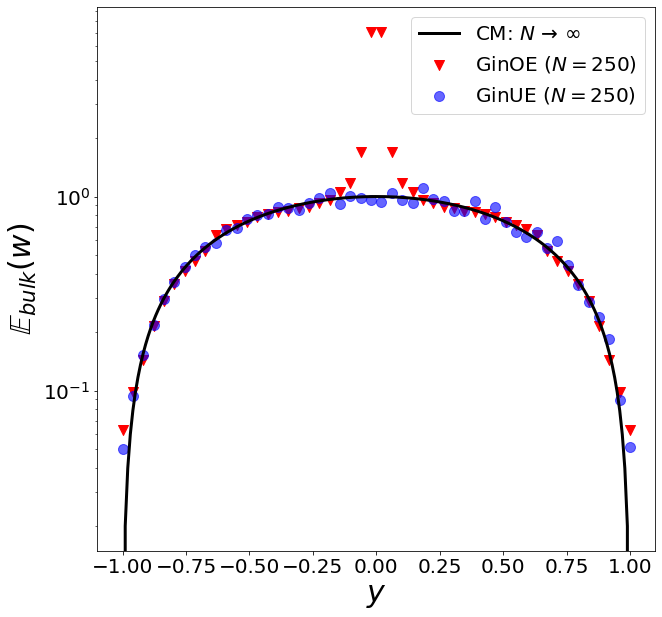}
    \caption{\small Numerical simulation of the mean conditional self-overlap, $\mathbb{E}_{\text{bulk}}(w)$, within the bulk of the GinOE and GinUE for large $N$ compared to the Chalker and Mehlig result (black lines). Left: $\mathbb{E}_{\text{bulk}}(w)$ in the bulk of GinOE as a function of $|w|$ for a range of different values of $N$ (coloured markers). We obtain results by considering the self-overlaps for eigenvalues within $\pm 1/\sqrt{N}$ of each value of $|z|$. Right: $\mathbb{E}_{\text{bulk}}(w)$ of eigenvectors of size $N=250$ associated with purely imaginary eigenvalues in the GinOE (red triangles) and GinUE (blue circles). Numerical results for each value of $N$ are taken from a sample with $O(10^6)$ GinOE and GinUE matrices. }     \label{fig:bulk_Gin}
\end{figure}

\noindent
One observes a crucial difference between the two ensembles when considering eigenvalues close to the real axis. This can be seen when measuring the mean self-overlap of eigenvectors associated with purely imaginary eigenvalues in the bulk for all $\text{Im}(z)\in [-\sqrt{N},\sqrt{N} ]$. In the case of the GinUE, the Chalker-Mehlig result holds for all $\text{Im}(z)$. In the GinOE however, the depletion of eigenvalues close to the real line leads to considerable deviations from the Chalker-Mehlig formula.  This effect will be accounted for by treating this region more carefully in Corollary \ref{cor:GinOEdepletionStrong}. However, before doing so, we consider the edge of the droplet and find the following Corollary, in agreement with its GinUE counterpart, given in Eq. \eqref{edgeasyGinUE}.

\begin{cor}\label{cor:GinOEedgeStrong}
    For a complex eigenvalue $z=\left( \sqrt{N} + \eta  \right)e^{i\theta}$ satisfying $\vert \sin \theta \vert \sim \mathcal{O}(1)$ and $\eta >0$ the limiting scaled mean self-overlap at the edge reads
    \be
    \label{Eq:GinOEedgeStrongRes}
        \mathcal{O}_{\textup{edge}}^{\textup{(GinOE,c)}}(\eta) \equiv \lim_{N\rightarrow \infty} \frac{1}{\sqrt{N}} \ \mathcal{O}_N^{\textup{(GinOE,c)}}\left( \left( \sqrt{N} + \eta  \right)e^{i\theta} \right)
        = \frac{1}{\pi} \left( \frac{1}{\sqrt{2\pi}} \ e^{-2\eta^2}  - \eta \ \textup{erfc}\left( \sqrt{2} \ \eta \right) \right) \ .
    \ee
\end{cor}

\noindent
We thus see that, away from the real axis, both the bulk asymptotic and the edge asymptotic of the mean diagonal overlap is shared between the GinOE and GinUE. Note that when approaching the boundary of the droplet, the mean self-overlap turns out to be parametrically weaker, which is reflected in rescaling with $1/\sqrt{N}$ instead of $1/N$ to obtain a non-trivial limit as $N\to \infty$. 

The result for the mean self-overlap at the spectral edge of the GinOE has been considered in Figure \ref{fig:edge_Gin}. One can see from this figure that as $N$ increases, the agreement between the theoretical and numerically observed mean self-overlaps becomes better.

\begin{figure}[h]
    \centering
    \includegraphics[scale = 0.215]{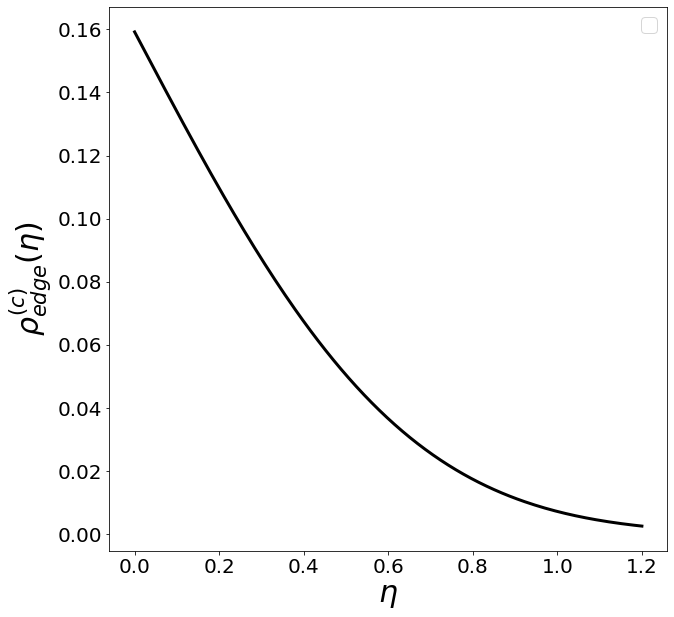}
    \hspace{2cm}
    \includegraphics[scale = 0.215]{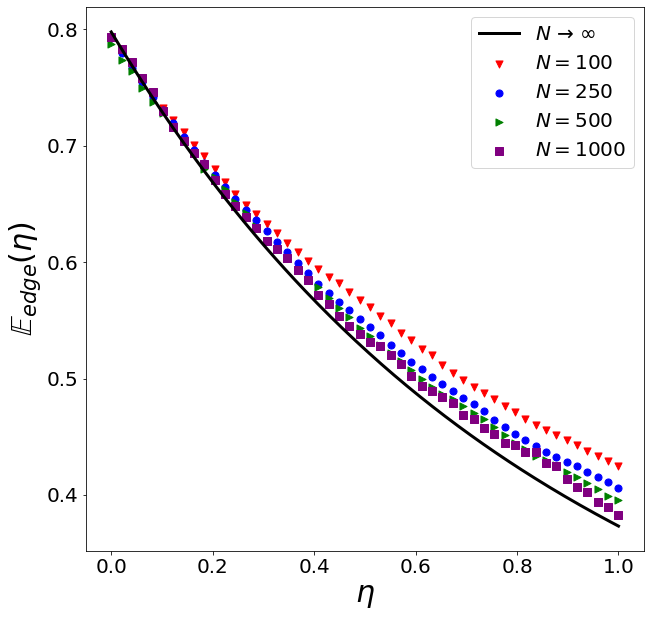}
    \caption{\small  Left: Limiting density of complex eigenvalues at the edge of the GinOE droplet. Right: Mean conditional self-overlap, $\mathbb{E}_{\text{edge}}(\eta)$, as a function of eigenvalue moduli, $|z| = \sqrt{N} + \eta$, for a range of different values of $N$. The limiting theoretical prediction of $\mathbb{E}_{\text{edge}}(\eta)$ for the GinOE (solid black line) is compared to numerical simulations for different values of $N$ (coloured markers). For each value of $\eta$, numerical averages are obtained by considering all eigenvalues with a modulus between $\pm 1/\sqrt{N}$ of $(\sqrt{N} + \eta)$. Samples of the self-overlap are taken from a data set generated from $O(10^6)$ GinOE matrices of each used value of $N$. }
    \label{fig:edge_Gin}
\end{figure}

\noindent
Finally, we present the asymptotic results for the depletion region of the droplet in the GinOE.

\begin{cor}\label{cor:GinOEdepletionStrong}
    For a complex eigenvalue $z=x+i\xi$, such that $\xi \sim \mathcal{O}(1)$ the limiting scaled mean self-overlap in the depleted region close to the origin, i.e. $x\sim \mathcal{O}(1)$, reads
    \be\label{Eq:GinOEdepletionStrongRes}
        \mathcal{O}^{\textup{(GinOE,c)}}_{\textup{depletion,origin}}(\xi) \equiv \lim_{N\rightarrow \infty} \frac{1}{N} \ \mathcal{O}_N^{\textup{(GinOE,c)}}\left( x+i\xi \right) 
        = \frac{1}{\pi} \left( 1+ \sqrt{\frac{\pi}{2}} \  \frac{1}{2|\xi|} \ e^{2\xi^2} \ \textup{erfc}\left( \sqrt{2} \ |\xi| \right) \right) \ .
    \ee
    Rescaling instead $x=\sqrt{N}\delta$, the limiting scaled mean self-overlap in the depleted region becomes
    \be\label{Eq:GinOEdepletionStrongRes2}
        \mathcal{O}^{\textup{(GinOE,c)}}_{\textup{depletion,strip}}(\delta,\xi) \equiv \lim_{N\rightarrow \infty} \frac{1}{N} \ \mathcal{O}_N^{\textup{(GinOE,c)}}\left( \sqrt{N} \delta +i\xi \right) 
        =  \mathcal{O}^{\textup{(GinOE,c)}}_{\textup{depletion,origin}}(\xi) \left( 1- \delta^2 \right) \Theta \left[ 1 - \delta^2 \right] \ .
    \ee
\end{cor}

\noindent
In order to demonstrate numerically the appropriate scale on which the depletion regime should be studied, the mean density of purely imaginary eigenvalues in the GinOE is plotted in Figure \ref{fig:depletion_Gin}. This illustrates a region of reduced eigenvalue density in the GinOE when $O(10^{-1}) < y < O(10^{1})$ before reaching an approximately constant value inside the bulk. Figure \ref{fig:depletion_Gin} also shows the mean conditional self-overlap, Eq. \eqref{conditional}, close to the origin and in a small strip close to the real line.

\begin{figure}[H]
    \centering
    \includegraphics[scale = 0.215]{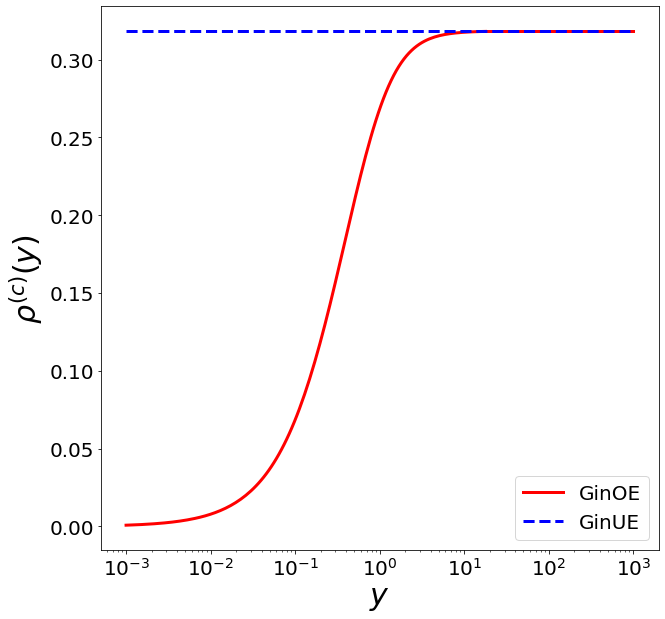}
    \hspace{0.05cm}
    \includegraphics[scale = 0.215]{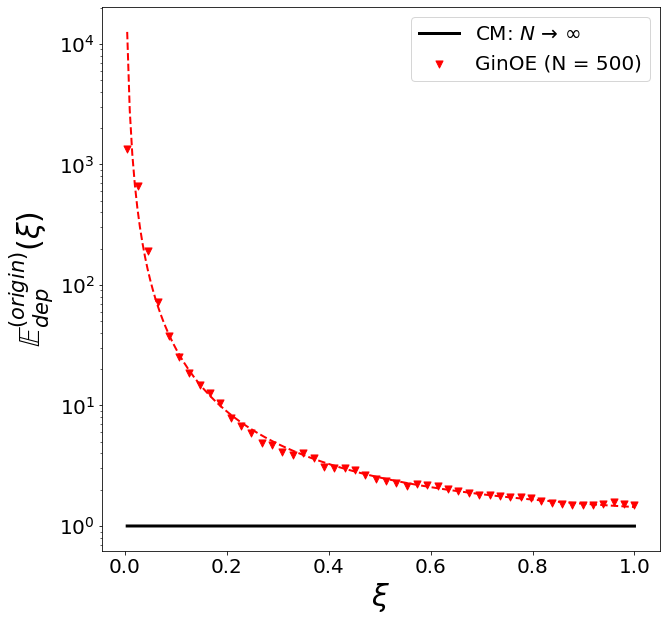}
    \hspace{0.05cm}
    \includegraphics[scale = 0.215]{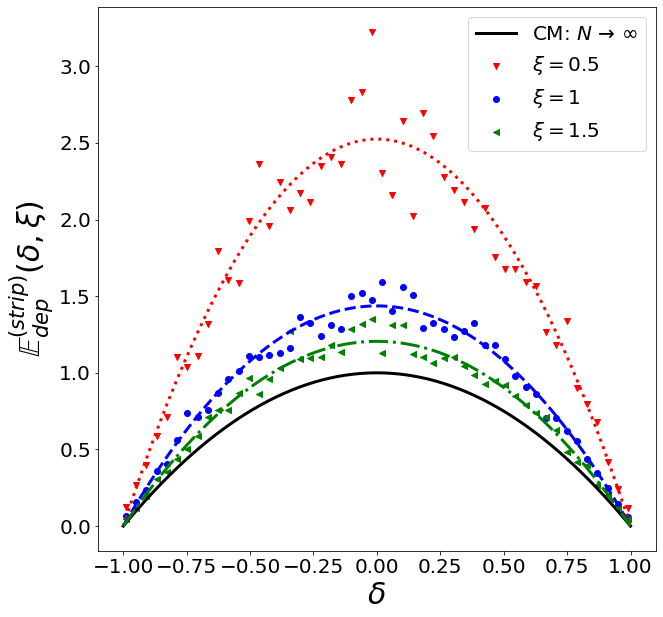}
    \caption{\small Large $N$ density of complex eigenvalues and mean conditional self-overlap $\mathbb{E}_{\text{dep}}(z)$ in the depletion regime of the GinOE. Left: Comparison of the limiting density of purely imaginary eigenvalues in the GinOE (red) and GinUE (blue). Centre: Theoretical prediction of $\mathbb{E}_{\text{dep}}^{\text{(origin)}}(\xi)$ close to the origin (red line), compared to numerical simulations (red triangles). Right: Mean conditional self-overlap $\mathbb{E}_{\text{dep}}^{\text{(strip)}}(\delta, \xi)$, with $z = \sqrt{N} \delta + i \xi$, for three different fixed values of $\xi$ in a strip close to the real line. In the centre and right hand plots the solid black line represents the Chalker and Mehlig result. Theoretical predictions (coloured lines) are compared to numerical simulations (coloured markers) of self-overlaps generated from $O(10^7)$ GinOE matrices of size $N=500$. }
    \label{fig:depletion_Gin}
\end{figure}

\noindent
When considering the mean self-overlap in this region, we start by considering eigenvalues close to the origin, i.e $z = x + i \xi$ with $x$ and $\xi \sim O(1)$. Here, one can see that Eq. \eqref{Eq:GinOEdepletionStrongRes} is independent of $x$ and depends purely on the imaginary component $\xi$. On the other hand, when considering the mean self-overlap along a rectangular strip close to the real line, i.e. taking $x \to \sqrt{N} \delta$, as in Eq. \eqref{Eq:GinOEdepletionStrongRes2}, there is now an explicit dependence on $\delta$. Essentially, one can interpret this expression as a scaled version of the Chalker-Mehlig result, that describes the increased mean self-overlap in the region of eigenvalue depletion close to the real line. Note that as $\xi$ becomes comparable to $\sqrt{N}$, i.e. the eigenvalue is inside the bulk, $\mathcal{O}^{\textup{(GinOE,c)}}_{\textup{depletion,origin}} (z) \to 1/\pi$ and we, unsurprisingly, recover the Chalker-Mehlig result.


\subsection{Numerical Simulations of the distribution of the diagonal overlap and discussion of open questions}
\label{subsec:Discussion}

We have already seen that, in the limit of large $N$, the mean self-overlap of eigenvectors in the spectral bulk and at the edge is the same for both the GinOE and GinUE. However, despite these similarities, there is a discernible difference in behaviour of the mean self-overlap in these two ensembles due to the existence of the depletion regime in the GinOE. It is natural to expect that a similar picture should hold not only for the first moment, but the whole distribution of the diagonal overlaps. As we do not yet have the analytic expression for such a distribution, $\mathcal{P}(t,z)$ where $t = \mathcal{O}_{nn} - 1$, for the GinOE in the complex plane, we proceed with briefly reviewing the results for complex eigenvalues of the GinUE and real eigenvalues of the GinOE, following the work \cite{FyodorovCMP}. The equation for the limiting JPDF of the eigenvector self-overlap in the bulk  of the GinUE reads:
\begin{align}
    \mathcal{P}^{\text{(GinUE,c)}}_{\text{bulk}} \left(s ,w \right) 
    = \frac{(1 - |w|^2)^2}{\pi s^3} e^{- \frac{1 - |w|^2}{s}} \Theta[ 1 - |w|^2] \ ,
    \label{eq:jpdf_GinOEc_bulk}
\end{align}
with $s = t/N$ and $z = \sqrt{N}w$, whereas the distribution at the spectral edge is given by
\begingroup
\allowdisplaybreaks
\begin{align}
    \mathcal{P}^{\text{(GinUE,c)}}_{\text{edge}} \left(\sigma , \eta \right)  
    =& \ \frac{1}{2\pi \sigma^5} e^{-\frac{\Delta^2}{2 \sigma^2}} \bigg\{  \frac{e^{-2\delta^2}}{\pi} \left( 2 \sigma^2 - \Delta \right) - \frac{1}{\sqrt{2\pi}} \left( 4 \delta \sigma^2 - \Delta(2 \delta + \sigma) \right) \erfc\left( \sqrt{2} \delta \right) \label{eq:jpdf_GinOEc_edge} \\
    & + \frac{e^{2 \delta^2}}{2} \left( \Delta^2 - \sigma^2 \right) \erfc^2\left( \sqrt{2} \delta \right) \bigg\} \nonumber \ ,
\end{align}
\endgroup
where $\sigma = t/ \sqrt{N}$ and $\Delta = 1 - 2 \sigma \eta$. In \cite{FyodorovCMP} one also finds an expression for the density of the self-overlap of GinOE eigenvectors associated with purely real eigenvalues in the bulk, $z=\sqrt{N}x$, which is given by
\begin{align}
    \mathcal{P}^{\text{(GinOE,r)}}_{\text{bulk}} \left(s, x \right) 
    = \frac{(1 - x^2)}{2 \sqrt{2\pi} } \frac{e^{- \frac{1 - x^2}{2s}}}{s^2} \Theta[ 1 - x^2] \ .
    \label{eq:jpdf_GinOEr_bulk}
\end{align}
\noindent
To compare with numerical simulations, the above distributions must also be normalised with respect to the mean spectral density. We denote the normalised distribution as $\widetilde{\mathcal{P}}(t,z) = \mathcal{P}(t,z)/\rho(z)$.

In Figure \ref{fig:distribution}, we compare distributions of the self-overlap of eigenvectors in the GinOE and GinUE. This is done in three different ways. Firstly, we consider the theoretical distributions in the complex bulk of the GinUE and real bulk of the GinOE (at $x=0$), in comparison to a numerically observed distribution in the depletion regime of the GinOE. We also compare the distribution of eigenvector self-overlaps in the complex bulk of the GinOE to the theoretical limiting distribution of self-overlaps in the bulk of the GinOE. Finally, we make a comparison between the distribution of the self-overlap at the edge of the GinUE and a numerically measured distribution at the edge of the GinOE. \\ 

\begin{figure}[h]
    \centering
    \includegraphics[scale = 0.215]{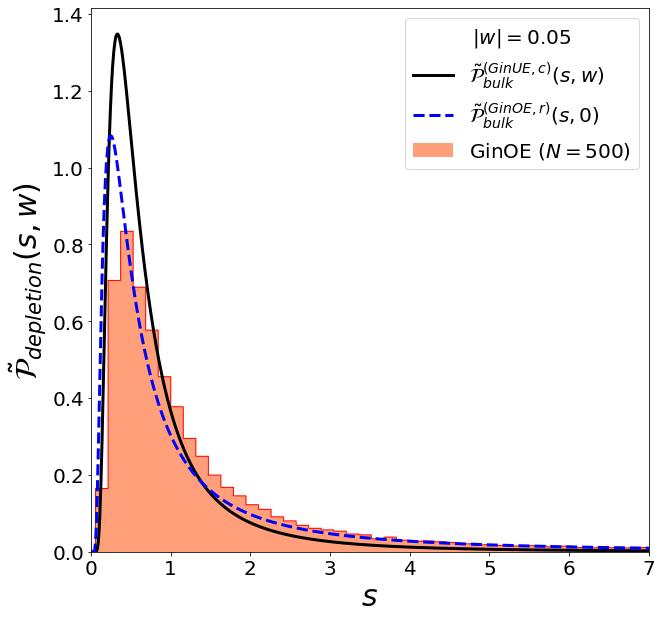}
    \hspace{0.1cm}
    \includegraphics[scale = 0.215]{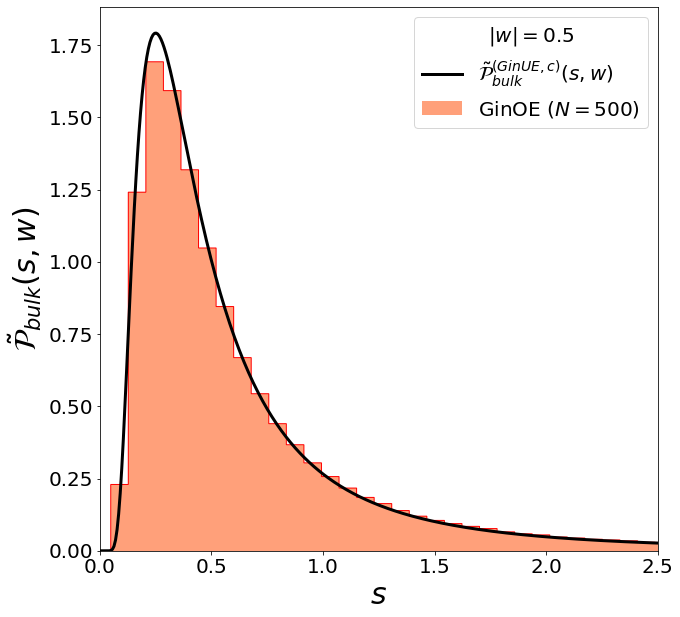}
    \hspace{0.1cm}
    \includegraphics[scale = 0.215]{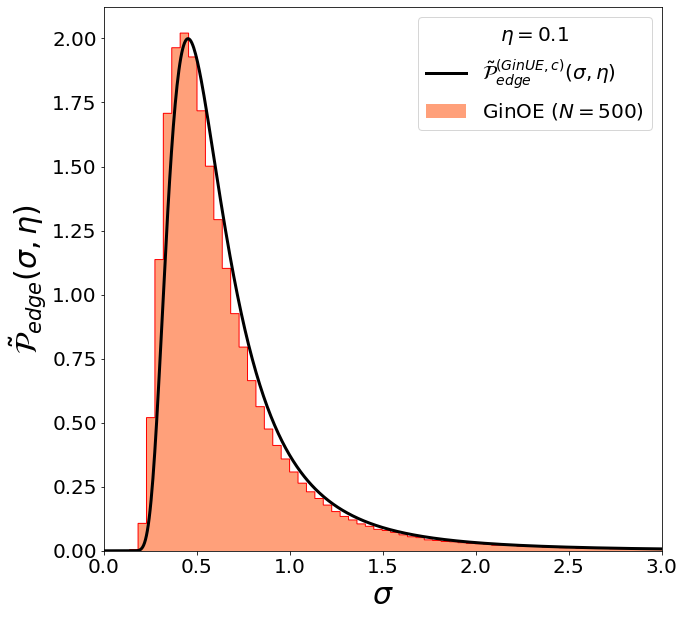}
    \caption{\small Numerical simulation of the distribution of the eigenvector self-overlaps in the GinOE compared to theoretical distributions for the GinOE and GinUE. In each plot the solid black line is the limiting distribution of the self-overlap in a region of the GinUE and the red shaded area shows a normalised histogram of self-overlaps observed in numerical simulation of $O(10^7)$ GinOE matrices of size $N=500$. Left: Depletion regime of the GinOE compared to the bulk of the GinUE and bulk real eigenvalues in the GinOE (blue line). Centre: Bulk of the GinOE compared to  the bulk of the GinUE. Right: Edge region of the GinOE compared to the edge of the GinUE. }
    \label{fig:distribution}
\end{figure}

\noindent
The data indeed seems to confirm that the GinUE distribution accurately describes the corresponding quantities for the GinOE away from the real axis.  In the depletion regime around the real axis the GinOE is not described by the limiting distribution associated with either complex eigenvalues in the bulk of the GinUE, Eq. \eqref{eq:jpdf_GinOEc_bulk}, or eigenvalues in the real bulk of the GinOE, Eq. \eqref{eq:jpdf_GinOEr_bulk}. It seems that the tails of the distribution of self-overlaps in the depletion regime is described better by Eq. \eqref{eq:jpdf_GinOEr_bulk} than Eq. \eqref{eq:jpdf_GinOEc_bulk}, indicating that the corresponding distribution may have a $1/s^{2}$ tail (at least as an intermediate asymptotics) as opposed to $1/s^{3}$. Next, considering the edge region it is apparent that there is a small discrepancy between theory and numerical simulations, however we attribute this to finite $N$ effects and expect this deviation to approach zero as $N \to \infty$.\\

\noindent
Rigorously proving equivalence of the distribution of the GinOE self-overlap for asymptotics within the bulk and at the edge for the GinOE remains an outstanding challenge. Characterizing the distribution of eigenvector diagonal overlaps in the depletion regime remains a completely open question. Another interesting extension of this work would be to study the off-diagonal averaged overlap, see Eq. \eqref{ChalkerMehligOverlap}, which is not yet known for the GinOE, neither for complex nor for real eigenvalues.\\

\noindent
In a separate paper, the results of this work will be extended to  the mean self-overlap in the complex plane of the real elliptic Ginibre ensemble, interpolating between GinOE and GOE.
In particular, this should allow us to study the depletion region in the new non-trivial scaling regime of weak non-Hermiticity \cite{FKS97a,FKS97b,FKS98}.


\section{Derivation of Main Results}\label{sec:MainResDerivation}

\noindent
Our strategy for the proof of the main results outlined in the previous section amounts to using the (incomplete) Schur decomposition with respect to pairs of complex eigenvalues of a real random matrix $G$. In this way we reduce the problem of evaluating the mean self-overlap in the complex plane to calculating the expectation value of a certain determinant, which is eventually implemented using Grassmann integration. In the process, we adapt techniques used in similar circumstances in \cite{FyodorovCMP} and \cite{FT}, which will help to prove the asymptotic results in Section \ref{sec:AsymptoticAnalysis}.


\subsection{Eigenvectors of real matrices via incomplete Schur decomposition}
\label{subsec:inSchur}

\noindent
Let $z=x+iy$ (with $y\neq 0$) be a complex eigenvalue of the $N\times N$ real, non-Hermitian matrix $G$. The associated left and right eigenvectors are denoted by $\mathbf{x}_L^\dagger$ and $\mathbf{x}_R$ respectively. Following the influential paper by Edelman \cite{Edelman}, we employ the \emph{incomplete Schur decomposition} of $G$ with respect to $z$, which is
\be\label{incomplSchur}
    G = Q \widetilde{G} Q^T \quad \text{with} \quad \widetilde{G} = \left(
    \begin{matrix}
    \begin{matrix}
    x & b \\
    -c & x \\
    \end{matrix} & W \\
    0 & G_2 \\
    \end{matrix}
    \right)\text{,} \ bc >0 \text{,}  \ b \geq c \text{,} \ y = \sqrt{bc}\text{,}
\ee
where real symmetric $Q$ is a so-called Householder reflection matrix  
such that $Q^2 = \eins_N$. The block matrix $W$ is real and has size $2\times (N-2)$, i.e. $W= \left(\mathbf{w}_1 \ \text{,} \ \mathbf{w}_2 \right)^T$, where $\mathbf{w}_1$, $\mathbf{w}_2$ are vectors with $N-2$ real i.i.d. entries following standard normal distributions. The matrix $G_2$ in this decomposition is essentially an $(N-2) \times (N-2)$ dimensional real Ginibre matrix. 
 The Jacobian of the transformation is presented in \cite{Edelman} and the integration measure changes according to
\be\label{JacSchur}
\begin{split}
    dG &= 2(b-c) \ \det \left[ (x\eins_{N-2} - G_2)^2 +y^2 \eins_{N-2} \right] \  dx \ db \ dc \ dW \ dG_2 \ dS \ ,
\end{split}
\ee
where $dS$ denotes the volume element of the Stiefel manifold originating from the Householder reflection matrix $Q$. After simple manipulations we find the probability measure, defined in terms of the new variables, reads:
\be\label{JPDFz}
\begin{split}
    P_{\text{GinOE}}(G)dG =& \ C_N^\prime \det \left[ (x\eins_{N-2} - G_2)^2 +y^2 \eins_{N-2} \right] e^{-\frac{1}{2}\Tr G_2 G_2^T}\\
    & \ (b-c) e^{-\frac{1}{2}\left( 2x^2 +b^2 +c^2 + \Tr WW^T \right)} dx \ db \ dc \ dW \ dG_2 \ ,
\end{split}
\ee
where the constant $C_{N}^\prime$ now reads
\be\label{Eq:NEWconsts}
\begin{split}
    C_{N}^\prime 
    = 2\ \frac{(2\pi)^{-\frac{1}{2}(N-1)^2}}{\sqrt{2\pi}\Gamma(N-1)} \ .
\end{split}
\ee
The next step is to determine the left and right eigenvectors in terms of the Schur decomposition variables, which is done adapting the method of \cite{FyodorovCMP}, where the case of real GinOE eigenvalues was considered. For $z=x+iy$, the eigenvalue problems read
\be\label{EigEq}
    G\mathbf{x}_R = z \mathbf{x}_R \quad \text{and} \quad \mathbf{x}_L^\dagger G = z \mathbf{x}_L^\dagger \ .
\ee
Applying the incomplete Schur decomposition from Eq. \eqref{incomplSchur}, we introduce $\widetilde{\mathbf{x}}_L^\dagger \equiv \mathbf{x}_L^\dagger Q$ and $\widetilde{\mathbf{x}}_R \equiv Q \mathbf{x}_R$ and see that the eigenvalue problems for $\widetilde{G}$ can be rewritten as $\widetilde{\mathbf{x}}_L^\dagger \widetilde{G} =z\widetilde{\mathbf{x}}_L^\dagger$
and $\widetilde{G}\widetilde{\mathbf{x}}_R= z \widetilde{\mathbf{x}}_R$.\\ 

\noindent
The left-right diagonal overlap, corresponding to the eigenvalue $z$ is obviously invariant under the incomplete Schur decomposition, i.e.
\be\label{invOverlap}
    \mathcal{O}_z = \left( \mathbf{x}_L^\dagger \mathbf{x}_L \right) \left( \mathbf{x}_R^\dagger \mathbf{x}_R \right)
    =\left(\widetilde{\mathbf{x}}_L^\dagger \widetilde{\mathbf{x}}_L \right) \left( \widetilde{\mathbf{x}}_R^\dagger \widetilde{\mathbf{x}}_R \right) \ ,
\ee
hence we can continue the calculation of the mean self-overlap using $\widetilde{\mathbf{x}}_L^\dagger$, $\widetilde{\mathbf{x}}_R$ and $\widetilde{G}$ instead of $\mathbf{x}_L^\dagger$, $\mathbf{x}_R$ and $G$.\\
\\
The incomplete Schur decomposition gives us the forms of $\widetilde{\mathbf{x}}_R$ and $\widetilde{\mathbf{x}}_L^\dagger$. By construction, $\widetilde{\mathbf{x}}_L^\dagger \widetilde{\mathbf{x}}_R = 1$ and it is easy to check that the eigenvectors must have the following structure:
\be
    \widetilde{\mathbf{x}}_R = \frac{1}{\sqrt{2}} \left( \begin{matrix}
    1 \\
    i\sqrt{\frac{c}{b}} \\
    \mathbf{0}_{N-2} \\
    \end{matrix}\right) \quad \text{and} \quad \widetilde{\mathbf{x}}_L^\dagger = \frac{1}{\sqrt{2}} \left( \begin{matrix}
    1 \\
    i\sqrt{\frac{b}{c}} \\
    \sqrt{2}\ \mathbf{b}_{N-2} \\
    \end{matrix}\right)^\dagger \ ,
\ee
with some $\mathbf{b}_{N-2}$ yet to be determined.   Substituting the above into the eigenvalue equation, $\widetilde{\mathbf{x}}_L^\dagger \widetilde{G} = z \widetilde{\mathbf{x}}_L^\dagger$, we find that
\be\label{bcond}
    z \mathbf{b}_{N-2}^\dagger \overset{!}{=} \frac{1}{\sqrt{2}} \left(\mathbf{w}_1^T - i \sqrt{\frac{b}{c}} \mathbf{w}_2^T \right) + \mathbf{b}_{N-2}^\dagger G_2 \quad \Leftrightarrow \quad \mathbf{b}_{N-2}^\dagger = \frac{1}{\sqrt{2}}  \left(\mathbf{w}_1^T - i \sqrt{\frac{b}{c}} \mathbf{w}_2^T \right) \left(z\eins_{N-2} - G_2 \right)^{-1} \ .
\ee
Using the above relation, the self-overlap $\mathcal{O}_z$ associated with a complex eigenvalue $z$ in the GinOE becomes
\be\label{overlapcomp}
\begin{split}
    \mathcal{O}_z &= \left( \widetilde{\mathbf{x}}_L^\dagger \widetilde{\mathbf{x}}_L \right) \left( \widetilde{\mathbf{x}}_R^\dagger \widetilde{\mathbf{x}}_R \right)
    =\frac{1}{4} \left( 2+ \frac{c^2+b^2}{bc} \right)+ \frac{1}{2} \left(1+\frac{c}{b} \right) \left( \mathbf{b}_{N-2}^\dagger \mathbf{b}_{N-2} \right)\ .
\end{split}
\ee

\noindent
To perform the ensemble average we follow \cite{Edelman} and change variables from $b$ and $c$ to $y=\sqrt{bc}$ and $\delta = b-c>0$, implying 
\be\label{Eq:bcqdchange}
\begin{split}
    db \ dc &= \frac{2y}{\sqrt{\delta^2 + 4y^2}} \ dy \ d\delta \ \text{,} \quad b^2 +c^2 = \delta^2 + 2y^2 \quad \text{and} \quad  \frac{b}{c} = \exp \left[\  2 \ \text{arcsinh}\left( \frac{\delta}{2y} \right) \right] \ \text{,}
\end{split}
\ee
so that the relevant JPDF, Eq. \eqref{JPDFz}, now takes the form
\be\label{Eq:JPDFs}
\begin{split}
     P_{\text{GinOE}}(G)dG =& \ C_{N}^\prime \det \left[ (x\eins_{N-2} - G_2)^2 +y^2 \eins_{N-2} \right] \exp \left[ -\frac{1}{2}\Tr \left( G_2 G_2^T + WW^T \right) \right] \\
    & \ \times \exp \left[ - x^2 \right] \ \exp \left[ - y^2 \right] \ \exp \left[  -\frac{1}{2} \delta^2  \right] \ \frac{2y\delta}{\sqrt{\delta^2 + 4y^2}} \ dx \ dy \ d\delta \ dW \ dG_2 \ .
\end{split}
\ee
Defining the matrix
\be\label{Eq:Bmat}
\begin{split}
    B &\equiv  \left(z\eins_{N-2} - G_2 \right)^\dagger \left(z\eins_{N-2} - G_2 \right) \ ,
\end{split}
\ee
we first express $\left( \mathbf{b}_{N-2}^\dagger \mathbf{b}_{N-2} \right)$ as 
\be\label{Eq:bbdaggerGinOE}
\begin{split}
      \left( \mathbf{b}_{N-2}^\dagger \mathbf{b}_{N-2} \right) =& \frac{1}{2} \mathbf{w}_1^T \ B^{-1} \ \mathbf{w}_1 +\frac{1}{2} \exp \left[\  2 \ \text{arcsinh}\left( \frac{\delta}{2y} \right) \right] \mathbf{w}_2^T \  B^{-1} \ \mathbf{w}_2 \\
      &+\frac{1}{2} \ i \ \exp \left[ \ \text{arcsinh}\left( \frac{\delta}{2y} \right) \right]  \bigg[ \mathbf{w}_1^T \  B^{-1} \ \mathbf{w}_2 - \mathbf{w}_2^T \  B^{-1} \ \mathbf{w}_1 \bigg] \ ,
\end{split}
\ee
which when substituted into Eq. \eqref{overlapcomp} implies
\be\label{overlapcomp2}
\begin{split}
    \mathcal{O}_z  = \widetilde{c}_1 + \widetilde{c}_2 \   \left( \mathbf{b}_{N-2}^\dagger \mathbf{b}_{N-2} \right) \ ,
\end{split}
\ee
where the constants read
\be\label{prefacs}
\begin{split}
    \widetilde{c}_1 &= \frac{1}{4} \left( 2+ \frac{\delta^2 + 2y^2}{y^2} \right)  \quad \text{and} \quad \widetilde{c}_2 =  \frac{1}{2} \left(1+ \exp \left[-2 \ \text{arcsinh}\left(\frac{\delta}{2y} \right) \right] \right) \ \text{.}
\end{split}
\ee
With these formulas in hand, we now proceed to proving our main theorem in the next section.


\subsection{Proof of Theorem \ref{thm:MainRes}}\label{subsec:ProoffiniteN}

\noindent
Using the equivalence of different indices $n$ under the ensemble averaging, the mean self-overlap for the GinOE can be written  as follows:
\be\label{Eq:Overlaps}
\begin{split}
    \mathcal{O}_N^{(\textup{GinOE},c)}(z) &= \bigg\langle \frac{1}{N}\sum_{n=1}^N \mathcal{O}_{nn} \ \delta(z-z_n) \bigg\rangle_{\text{GinOE},N} = \bigg\langle \mathcal{O}_{\widetilde{z}} \ \delta(z-\widetilde{z}) \bigg\rangle_{\textup{GinOE},N} \ \text{,}
\end{split}
\ee
where the two-dimensional $\delta$-function of complex argument, $z=x+iy$, should be interpreted as the product of two one-dimensional $\delta$-functions, containing its real and imaginary parts respectively, i.e. $\delta(z-\widetilde{z})=\delta(x-\widetilde{x})\delta(y-\widetilde{y})$. This allows us to use the incomplete Schur decomposition as presented in the previous section, in particular using results from Eq. \eqref{Eq:JPDFs} to Eq. \eqref{prefacs}. The average splits into integrations over $G_2$, $W$, $x$, $y$ and $\delta$ according to Eq. \eqref{Eq:JPDFs}.  The two $\delta$-functions make the integrations over $x$ and $y$ trivial, meaning that the next non-trivial task is to perform the integration with respect to the matrix $W$. We have $\Tr WW^T = \mathbf{w}_1^T \mathbf{w}_1 + \mathbf{w}_2^T \mathbf{w}_2$ and $dW = d\mathbf{w}_1 d\mathbf{w}_2$, which allows us to perform the Gaussian averages with respect to vectors $\mathbf{w}_1$ and $\mathbf{w}_2$ instead of $W$. We define the average with respect to $\mathbf{w}$ of an observable $\mathcal{A}(\mathbf{w})$ as
\be\label{Eq:WaverageDef}
    \bigg\langle \mathcal{A}(\mathbf{w}) \bigg\rangle_{\mathbf{w}} \equiv  \frac{1}{(2\pi)^{\frac{N-2}{2}}} \int d\mathbf{w} \exp \left[ -\frac{1}{2} \mathbf{w}^T \mathbf{w} \right] \mathcal{A}(\mathbf{w})\ \text{,}
\ee
 normalized in such a way that $\langle \eins \rangle_{\mathbf{w}} = 1$. 
The mean self-overlap for the GinOE then can be written as
\be\label{Eq:avOvlGinOEstep}
\begin{split}
    &\bigg\langle \mathcal{O}_{\widetilde{z}} \ \delta(z-\widetilde{z}) \ \bigg\rangle_{\textup{GinOE},N} = C_{N}^\prime \ (2\pi)^{N-2} \ \exp \left[ - \left( x^2 + y^2 \right) \right] \int d\delta \ \frac{2y\delta}{\sqrt{\delta^2 + 4y^2}} \ \exp \left[  -\frac{1}{2} \delta^2  \right] \\
    &\times  \int dG_2 \ \det \left[ (x\eins_{N-2} - G_2)^2 +y^2 \eins_{N-2} \right] \exp \left[ -\frac{1}{2}\Tr G_2 G_2^T \right] \ \bigg\langle \bigg\langle \widetilde{c}_1 + \widetilde{c}_2 \   \left( \mathbf{b}_{N-2}^\dagger \mathbf{b}_{N-2} \right) \bigg\rangle_{\mathbf{w}_2} \bigg\rangle_{\mathbf{w}_1} \ .
\end{split}
\ee
The computation of the double average over $\mathbf{w}_1$, $\mathbf{w}_2$ is performed in the next step, using  the following Lemma.

\begin{lem}\label{lem:Wav1}
Let $\mathbf{w}_1$, $\mathbf{w}_2$ be two vectors, each of length $N-2$, with independent real variables as entries and $X$ be an $N-2$ dimensional matrix. Then
\be\label{Eq:LemWaverage1}
\begin{split}
    &\bigg\langle \mathbf{w}^T \  X \ \mathbf{w} \bigg\rangle_{\mathbf{w}} = \Tr X \quad \textup{and} \quad \bigg\langle \bigg\langle    \mathbf{w}_1^T \  X \  \mathbf{w}_2  \bigg\rangle_{\mathbf{w}_2} \bigg\rangle_{\mathbf{w}_1} = \bigg\langle \bigg\langle     \mathbf{w}_2^T \  X \  \mathbf{w}_1 \bigg\rangle_{\mathbf{w}_2} \bigg\rangle_{\mathbf{w}_1} = 0 \ \text{,}
\end{split}
\ee
with the average defined in Eq. \eqref{Eq:WaverageDef} and $\mathbf{w} \in \{ \mathbf{w}_1, \mathbf{w_2} \}$. 
\end{lem}

\noindent
The verification of this Lemma is straightforward and can be done by exploiting the Gaussian nature of the integrals involved. This allows us to show the validity of the next Lemma.
\begin{lem}\label{lem:Wav2}
    The average over $\mathbf{w}_1$ and $\mathbf{w}_2$ of the object $\left( \mathbf{b}_{N-2}^\dagger \mathbf{b}_{N-2} \right)$ reads
    \be\label{Eq:LemWaverage3}
    \begin{split}
        \bigg\langle \bigg\langle \widetilde{c}_1 + \widetilde{c}_2 \   \left( \mathbf{b}_{N-2}^\dagger \mathbf{b}_{N-2} \right)  \bigg\rangle_{\mathbf{w}_2} \bigg\rangle_{\mathbf{w}_1}
        &= \frac{1}{2} \left(2 + \frac{\delta^2}{2y^2}  \right) \bigg[ 1 + \Tr \left[ B^{-1} \right] \bigg] \ .
    \end{split}
    \ee
\end{lem}
\begin{proof}
    The verification amounts to using the expression of $\left( \mathbf{b}_{N-2}^\dagger \mathbf{b}_{N-2} \right)$, in Eq. \eqref{Eq:bbdaggerGinOE}, in conjunction with Lemma \ref{lem:Wav1}, such that:  
    \be\label{avswsell}
    \begin{split}
        \bigg\langle \ \bigg\langle  \left( \mathbf{b}_{N-2}^\dagger \mathbf{b}_{N-2} \right) \bigg\rangle_{\mathbf{w}_2} \bigg\rangle_{\mathbf{w}_1} 
        &=\frac{1}{2} \bigg\langle \mathbf{w}_1^T \ B^{-1} \ \mathbf{w}_1 \bigg\rangle_{\mathbf{w}_1} +\frac{1}{2} \exp \left[\  2 \ \text{arcsinh}\left( \frac{\delta}{2y} \right) \right] \bigg\langle \mathbf{w}_2^T \  B^{-1} \ \mathbf{w}_2 \bigg\rangle_{\mathbf{w}_2}\\
        &= \frac{1}{2} \left( 1 + \exp \left[\  2 \ \text{arcsinh}\left( \frac{\delta}{2y} \right) \right] \right) \Tr B^{-1} \ .
    \end{split}
    \ee 
    Then exploiting the fact that $ \cosh^2 \left[\text{arcsinh}\left(\frac{\delta}{2y} \right) \right] = 1 + \frac{\delta^2}{4y^2}$,
     alongside the definition of $\Tilde{c}_2$ in Eq. \eqref{prefacs} this yields: 
    \be
    \begin{split}
         \frac{\widetilde{c}_2}{2} \left( 1 + \exp \left[\  2 \ \text{arcsinh}\left( \frac{\delta}{2y} \right) \right] \right) 
         = \frac{1}{2} \bigg[ 1 + \cosh \left[2 \ \text{arcsinh}\left(\frac{\delta}{2y} \right) \right] \bigg] =   1+\frac{\delta^2}{4y^2} .
    \end{split}
    \ee
    Finally, applying the definition of $\Tilde{c}_1$ from Eq. \eqref{prefacs}, leads to
    \be
    \begin{split}
        \bigg\langle \bigg\langle \widetilde{c}_1 + \widetilde{c}_2 \   \left( \mathbf{b}_{N-2}^\dagger \mathbf{b}_{N-2} \right)  \bigg\rangle_{\mathbf{w}_2} \bigg\rangle_{\mathbf{w}_1} &= \widetilde{c}_1+ \frac{1}{2} \Tr \left[ B^{-1} \right]  \left(2 + \frac{\delta^2}{2y^2}  \right) 
        = \frac{1}{2} \left(2 + \frac{\delta^2}{2y^2}  \right) \bigg[ 1 + \Tr \left[ B^{-1} \right] \bigg] \ ,
    \end{split}
    \ee
    which concludes the proof of this Lemma.
\end{proof}


\noindent
With Lemma \ref{lem:Wav2}, we are able to express the mean self-overlap at finite $N$ in such a way, that the remaining average over $G_2$ becomes tractable. Starting from Eq. \eqref{Eq:avOvlGinOEstep} and applying Lemma \ref{lem:Wav2} we obtain
\be\label{Eq:avOvlGinOEstep2} 
\begin{split}
    \bigg\langle \mathcal{O}_{\widetilde{z}} \ \delta(z-\widetilde{z}) \ \bigg\rangle_{\textup{GinOE},N} 
    &= C_{N}^\prime \ (2\pi)^{N-2} \ \exp \left[ - \left( x^2 + y^2 \right) \right] \int d\delta \ \frac{2y\delta}{\sqrt{\delta^2 + 4y^2}} \ \exp \left[  -\frac{1}{2} \delta^2  \right]  \frac{1}{2} \left(2 + \frac{\delta^2}{2y^2}  \right) \\
    &\times  \int dG_2 \ \det \left[ (x\eins_{N-2} - G_2)^2 +y^2 \eins_{N-2} \right] \exp \left[ -\frac{1}{2}\Tr G_2 G_2^T \right] \bigg[ 1 + \Tr \left[ B^{-1} \right] \bigg] \ .
\end{split}
\ee
The integrand, involving the matrix $G_2$, can be written in terms of the matrix $B$ from Eq. \eqref{Eq:Bmat}, noticing that
\be\label{detformulas}
\begin{split}
    \det B &= \det \left[ \left(z\eins_{N-2} - G_2 \right)^\dagger \left(z\eins_{N-2} - G_2 \right)  \right] = \det \left[ (x\eins_{N-2} - G_2)^2 +y^2 \eins_{N-2} \right] \\
    &= \det \left[ \begin{matrix}
    0 & i(z\eins_{N-2} - G_2)  \\
    i(\bar{z} \eins_{N-2} - G_2^T) & 0 \\
    \end{matrix}  \right] \ ,
\end{split}
\ee
as well as
\be\label{deriv}
    \frac{\partial}{\partial \mu} \det\left( \mu \eins_{N-2} + B \right) \bigg\vert_{\mu=0} = \det B \ \Tr B^{-1} \ .
\ee
In fact, it is more convenient to introduce a block-matrix $M$ via
\be\label{Eq:propMmat}
    M \equiv \left( \begin{matrix}
        \sqrt{\mu} \eins_{N} & i\left( z \eins_N - G \right) \\
        i\left( \bar{z} \eins_N - G^T \right) & \sqrt{\mu} \eins_N \\
    \end{matrix} \right) \ ,
\ee
where each block is of size $N\times N$ and $\mu$ is a real parameter. It is easy to see that
for $\mu=0$ we then have
\be
    \det M = \det B  \quad \text{and} \quad \frac{\partial}{\partial \mu} \det M \ \bigg\vert_{\mu=0} = \det B \ \Tr B^{-1} \ .
\ee
Now we define the averaging over $G_2$  as
\be\label{Eq:X2av}
    \bigg\langle \mathcal{A}(G_2) \bigg\rangle_{G_2} \equiv C_{N-2}^{-1} \int dG_2 \ \exp \left[-\frac{1}{2}\Tr \left( G_2 G_2^T  \right) \right] \ \mathcal{A}(G_2) \ ,
\ee
where the constant $C_{N-2}$ is given in Eq. \eqref{GinDistrib}  to ensure the correct normalization: $\langle \eins \rangle_{G_2} = 1$. The expression in Eq. \eqref{Eq:avOvlGinOEstep2} therefore becomes
\be\label{Eq:avOvlGinOEstep3}
\begin{split}
    \bigg\langle \mathcal{O}_{\widetilde{z}} \ \delta(z-\widetilde{z}) \ \bigg\rangle_{\textup{GinOE},N} 
    &= C_{N}^\prime \  (2\pi)^{N-2} \ C_{N-2} \ \exp \left[ - \left( x^2 + y^2 \right) \right] \int d\delta \ \frac{y\delta}{\sqrt{\delta^2 + 4y^2}} \ \exp \left[  -\frac{1}{2} \delta^2  \right]   \left(2 + \frac{\delta^2}{2y^2}  \right)\\
    &\times \bigg[ \ \bigg\langle \det M \bigg\rangle_{G_2} \ \bigg\vert_{\mu = 0} + \frac{\partial}{\partial \mu }\bigg\langle \det M  \bigg\rangle_{G_2} \ \bigg\vert_{\mu = 0} \ \bigg] \ .
\end{split}
\ee


\noindent
The analysis of the average of the determinant of the matrix $M$ with respect to the GinOE of size $N-2$ is  done using the standard representation of the determinant via Berezin integration over anti-commuting Grassmann variables \cite{Berezin}:
\be\label{detIdentGrass}
    \int D(\bm \Phi,\bm \chi) \exp \left[ - \bm \Phi^T M \bm \chi \right] = \det M \ \text{,}
\ee
where $D(\bm \Phi,\bm \chi) = d\bm \phi_1 d \bm \chi_1 d \bm \phi_2 d\bm \chi_2$ and $\bm \Phi^T = \left( \bm \phi_1, \bm \phi_2 \right)^T$, $\bm \chi^T = \left( \bm \chi_1, \bm \chi_2 \right)^T$.
The outcome is provided by the following Proposition.
\begin{prop}\label{prop:G2av}
    Let $M$ be the matrix defined in Eq. \eqref{Eq:propMmat}, 
     $z=x+iy$ is a complex number and $G$ is an $N\times N$ real Ginibre matrix. The average over the determinant of $M$, with respect to the real Ginibre matrix $G$, can be expressed in the following form:
    \begin{align}
        \label{Eq:detMresGinOEa}
        &\bigg\langle \det M \bigg\rangle_{G} = \frac{1}{\pi}  \int_0^{\infty} dr ~ r \ e^{-r^2} \int_0^{2\pi} d \theta \bigg[  \mu + |z|^2 + 2\sqrt{\mu} r \cos(\theta) + r^2 \bigg]^N 
         \\
        \label{Eq:detMresGinOEb} &=   \int_0^{\infty} dR  \ e^{-R}  \bigg[  \left(R + |z|^2\right)^2 + \mu^2 +2\mu (|z|^2-R)  \bigg]^{\frac{N}{2}} 
        \,P_N\left(\frac{\mu+R+|z|^2}{\sqrt{\left(R + |z|^2\right)^2 + \mu^2 +2\mu (|z|^2-R)}}\right) \ ,
    \end{align}
    where $P_N(t) $ is a Legendre polynomial defined via the identity \textup{\cite[8.913.3]{Grad}}
    \be \label{LegPolyint}
        P_N(t)=\frac{1}{\pi}\int_0^{\pi} d\theta \ \left(t+\sqrt{t^2-1}\cos{\theta}\right)^N \ .
    \ee
\end{prop}
\begin{proof}
Using the integration over vectors with $N$ anticommuting components each, $\bm \phi_1$, $\bm \phi_2$, $\bm \chi_1$, $\bm \chi_2$,  with entries $\chi_{1,i}$ $(i = 1,2, ... N)$ and so forth, we start with writing Eq. \eqref{detIdentGrass} in the explicit form for our particular choice:
\be
   \bigg\langle \det M \bigg\rangle_{G} =(-1)^N \expval{ \int D \bm \Phi D \bm \chi 
    \exp{ -  \left( \begin{matrix} \bm \phi_1^T & \bm \phi_2^T \end{matrix} \right)
    \left( \begin{matrix} \sqrt{\mu} \eins_N & i \left( z \eins_N -  G \right) \\ i \left( \bar{z} \eins_N -  G^T \right) & \sqrt{\mu} \eins_N
    \end{matrix} \right) \left( \begin{matrix} \bm \chi_1 \\ \bm \chi_2 \end{matrix} \right)
     }  }_{G} \ ,
\ee
which upon using $\bm \phi_1^T G \bm \chi_2 =- \Tr[G \bm \chi_2 \otimes \phi_1^T  ]$ can be written as
\begin{align}
    \bigg\langle \det M \bigg\rangle_{G} 
     = (-1)^N \int& D \bm \Phi D \bm \chi e^{ - \sqrt{\mu} \bm \phi_1^T \bm \chi_1  - \sqrt{\mu} \bm \phi_2^T \bm \chi_2 - i z \bm \phi_1^T \bm \chi_2 - i \bar{z} \bm \phi_2^T \bm \chi_1 }  \expval{ 
     e^{ - i \Tr[G  \bm \chi_2 \otimes \bm \phi_1^T ] - i \Tr[ G^T \bm \chi_1 \otimes \bm \phi_2^T] }  }_{G}  \ .
\end{align}
The expectation value in the above integrand can be evaluated using the following identity
\begin{equation}
    \expval{e^{-\Tr[G  \bm A] - \Tr[ G^T \bm B]}}_{G} = \exp{ \frac{1}{2} \Tr[\bm A^T \bm A] + \frac{1}{2} \Tr[\bm B^T \bm B] + \Tr[\bm A \bm B] } \ ,
    \label{eq:ave_exp_Tr_GN}
\end{equation}
which can be found in \cite[Eq. (3.15)]{FyodorovCMP}. In our case we have that $\bm A =  i\bm \chi_2 \otimes \bm \phi_1^T$ and $\bm B = i\bm \chi_1 \otimes \bm \phi_2^T$, so that
\begin{equation}
    \Tr[\bm A^T \bm A] = \Tr[\bm B^T \bm B] = 0 \quad \text{and} \quad \Tr[\bm A \bm B] = \left(\bm \phi_1^T \bm \chi_1 \right)\left(\bm \phi_2^T \bm \chi_2\right) \ ,
\end{equation}
yielding
\be
    \bigg\langle \det M \bigg\rangle_{G} =  (-1)^N
    \int D \bm \Phi D \bm \chi \exp{ - \sqrt{\mu} \bm \phi_1^T \bm \chi_1  - \sqrt{\mu} \bm \phi_2^T \bm \chi_2 - i z \bm \phi_1^T \bm \chi_2 - i \bar{z} \bm \phi_2^T \bm \chi_1 +\left(\bm \phi_1^T \bm \chi_1 \right)\left(\bm \phi_2^T \bm \chi_2\right) } \ .
\ee
The exponential of the term  quartic in anticommuting variables can be re-expressed using a Hubbard-Stratonovich transformation of the form
\begin{equation}
    e^{ab} = \frac{1}{2\pi} \int d\bar{q} dq e^{-\vert q \vert^2 - (aq + b\bar{q})} \ ,
\end{equation}
so that after changing the order of integration and  performing the Gaussian integrals over anticommuting vectors we arrive at
\begingroup
\allowdisplaybreaks
\begin{align}
    \bigg\langle \det M \bigg\rangle_{G} &= \frac{(-1)^N}{2\pi} \int d\bar{q} dq e^{-\vert q \vert^2}
     \int D \bm \Phi D \bm \chi e^{ - \sqrt{\mu} \bm \phi_1^T \bm \chi_1  - \sqrt{\mu} \bm \phi_2^T \bm \chi_2 - i z \bm \phi_1^T \bm \chi_2 - i \bar{z} \bm \phi_2^T \bm \chi_1 - q \bm \phi_1^T \bm \chi_1 - \bar{q} \bm \phi_2^T \bm \chi_2} \\
    &= \  \frac{1}{2\pi} \int d\bar{q} dq e^{-\vert q \vert^2}
     \left[  \left( \sqrt{\mu} + q \right)\left( \sqrt{\mu} + \bar{q} \right) + |z|^2 \right]^N \ .\label{res05}
\end{align}
\endgroup
      Employing the change of variables to polar coordinates, $q=re^{i \theta}$ and $\bar{q}=re^{-i \theta}$, yields Eq. \eqref{Eq:detMresGinOEa}, then using $r^2=R$ and the definition of a Legendre polynomial, Eq. \eqref{LegPolyint}, yields the second form, Eq. \eqref{Eq:detMresGinOEb}.
\end{proof}

\begin{rem}\label{rem:Kernel}
    Proposition \ref{prop:G2av} allows one to compute the two terms in Eq. \eqref{Eq:avOvlGinOEstep3} which require averages with respect to $G_2$. Note that for $\mu=0$ the determinant of the matrix $M$ can be written as a product of two characteristic polynomials of Ginibre matrices. 
    The average of a product of two characteristic polynomials is proportional to the associated kernel at equal arguments for complex eigenvalues in the GinOE, and is well-known, see e.g. \cite[Eq. (18.5.40)]{KS}. However, we need a slightly more general average involving the derivative over $\mu$. We state the results we need in the next Corollary.
\end{rem}

\begin{cor}\label{cor:AvdetG2}
    With the average taken as in Eq. \eqref{Eq:X2av} using the $N-2$ sized GinOE matrix $G_2$, we have 
    \be\label{Eq:AvdetMmu0}
        \bigg\langle \det M \bigg\rangle_{G_2} \ \bigg\vert_{\mu=0} =  e^{\vert z \vert^2} \ \Gamma\left(N-1,\vert z \vert^2 \right)
    \ee
    and
    \be\label{Eq:AvdetMpartialmu0}
        \frac{\partial}{\partial \mu} \bigg\langle \det M \bigg\rangle_{G_2} \ \bigg\vert_{\mu=0} =  \big( N - 2 - |z|^2 \big) \ e^{|z|^2} \ \Gamma(N - 1 , |z|^2) + |z|^{2(N-1)} \ .
    \ee
\end{cor}
\noindent
where we have used the incomplete Gamma-function defined in Eq. \eqref{Eq:incompleteGamma}.
\begin{proof}
    The proof of Eq. \eqref{Eq:AvdetMmu0} is immediate from Eq. \eqref{Eq:detMresGinOEb} after changing $N\to N-2$, using $P_{N-2}(1)=1$ \cite[Table 18.6.1]{NIST} and the identity \cite[8.6.5]{NIST}
    \be\label{Gamgam}
        \int_0^\infty\,e^{-R}(R+|z|^2)^{N-2}\,dR=e^{|z|^2}\Gamma(N-1,|z|^2) \ .
    \ee
    In order to prove Eq. \eqref{Eq:AvdetMpartialmu0}, we replace $N\to N-2$ in  Eq. \eqref{Eq:detMresGinOEb}, then
    differentiate over $\mu$ using
    \be
        \frac{\partial}{\partial \mu} \left(\frac{\mu+R+|z|^2}{\sqrt{\left(R + |z|^2\right)^2 + \mu^2 +2\mu (|z|^2-R)}}\right)\bigg\vert_{\mu=0}
        =\frac{2R}{(R+|z|^2)^2} 
    \ee
    and also $P_{N-2}'(1)=(N-2)(N-1)/2$, which follows per induction from \cite[8.939.6]{Grad}. Collecting all terms gives
    \be\label{intermed}
    \begin{split}
        &\frac{\partial}{\partial \mu} \bigg\langle \det M \bigg\rangle_{G_2} \ \bigg\vert_{\mu=0} = (N-2)\int_0^\infty \,dR \ 
        e^{-R} \ (R+|z|^2)^{N-4} \ \bigg[|z|^2+R(N-2)\bigg]\\
        &=(N-2)\left[(N-2)\int_0^\infty \,dR \ 
        e^{-R}\left(R+|z|^2\right)^{N-3}-(N-3) \ |z|^2 \ \int_0^\infty \,dR \ 
        e^{-R}\left(R+|z|^2\right)^{N-4}\right]\\
        &=(N-2)\ e^{|z|^2}\bigg[(N-2) \ \Gamma(N-2,|z|^2)-(N-3)\ |z|^2 \ \Gamma(N-3,|z|^2)\bigg] \ ,
    \end{split}
    \ee
    where we have again utilised Eq. \eqref{Gamgam}. By now substituting the following identity for incomplete $\Gamma$-functions, see e.g. \cite[8.356.2]{Grad},
    \be\label{GammaRel}
       m \ \Gamma\left(m,x\right) =\Gamma\left(m+1,x\right) - e^{-x} x^{m} \ , 
    \ee
    in Eq. \eqref{intermed}, both for $m=N-2$ and $m=N-3$, one finds that
    \be 
         \frac{\partial}{\partial \mu} \bigg\langle \det M \bigg\rangle_{G_2} \ \bigg\vert_{\mu=0} = \big( N - 2 - |z|^2 \big) \ e^{|z|^2} \ \Gamma(N - 1 , |z|^2) + |z|^{2(N-1)} \ , 
    \ee
    which concludes our proof of Eq. \eqref{Eq:AvdetMpartialmu0}.
\end{proof}

\noindent
We can now finish the proof of Theorem \ref{thm:MainRes}. Going back to Eq. \eqref{Eq:avOvlGinOEstep3}
and utilising the results from Corollary \ref{cor:AvdetG2} for the averages over $G_2$, we get the following expression:
\be\label{Eq:avOvlGinOEstep5}
\begin{split}
     \bigg\langle \mathcal{O}_{\widetilde{z}} \ \delta(z-\widetilde{z}) \bigg\rangle_{\textup{GinOE},N} 
     =& \ C_{N}^\prime \ (2\pi)^{N-2} \ C_{N-2} \ \ \frac{1}{2y} \int d\delta \ \delta \ \sqrt{\delta^2 + 4y^2} \ \exp \left[  -\frac{1}{2} \delta^2  \right]   \\
     & \times \bigg[ \Gamma(N - 1 , |z|^2) \  \big( N - 1 - |z|^2 \big)  + |z|^{2(N-1)} e^{-|z|^2}\bigg] \ .
\end{split}
\ee
Finally, calculating the remaining integral over $\delta$, using \cite[3.382.4]{Grad}, Eq. \eqref{GammaRel} and \cite[8.4.6]{NIST}, as
\be\label{eq:deltaInt}
    \frac{1}{2y} \ \int_0^\infty d\delta \ \delta \  \sqrt{\delta^2 + 4y^2} \ e^{-\frac{1}{2}\delta^2} = 1 + \sqrt{\frac{\pi}{2}} \ e^{2y^2} \ \frac{1}{2\vert y \vert } \ \text{erfc}\left(\sqrt{2} \ \vert y \vert \right)    
\ee
and collecting all multiplicative constants together yields
\be\label{Eq:MainResGinOE_text}
    \begin{split}
        \mathcal{O}^{(\GinOE,c)}_{N}(z) &= \bigg\langle \frac{1}{N}\sum_{n=1}^N \mathcal{O}_{nn} \ \delta(z-z_n) \bigg\rangle_{\textup{GinOE},N}  = \frac{1}{\pi} \ \left( 1 + \sqrt{\frac{\pi}{2}} \ \exp \left[ 2y^2 \right] \ \frac{1}{2 \vert y \vert} \ \textup{erfc}\left(\sqrt{2} \ \vert y \vert \right) \right) \\
        &\times \bigg[  \ \frac{\Gamma\left(N-1,\vert z \vert^2 \right)}{(N-2)!}  \bigg[ N - 1 - \vert z \vert^2 \bigg] +  \frac{\vert z\vert^{2(N-1)}}{(N-2)!}  \ e^{-\vert z \vert^2} \bigg] \ ,
    \end{split}
\ee
thus proving the Theorem \ref{thm:MainRes}.


\section{Asymptotic Analysis for large matrix size}\label{sec:AsymptoticAnalysis}

\noindent
In this Section, we prove the large $N\gg 1$ asymptotic results in various scaling regimes, which are stated in Corollaries \ref{cor:GinOEbulkStrong}, \ref{cor:GinOEedgeStrong} and \ref{cor:GinOEdepletionStrong}. The starting point is always the finite $N$ result in Theorem \ref{thm:MainRes}, from there we then perform our asymptotic analysis via Laplace's method.

\begin{rem}\label{rem:erfcAsymptotic}
    The asymptotic analysis at $N\rightarrow \infty$ becomes simpler, after defining the function
    \be\label{Thetafunction}
        \Theta_N^{(M)}(x) \equiv \frac{\Gamma(N-M+1,Nx)}{\Gamma(N-M+1)} \ ,
    \ee
    which appears in the expression at finite $N$ for the mean self-overlap in the GinOE. We observe that, for real $x$ and fixed $M$, this function is bounded for all $N$ by $1$, i.e. $\Theta_N^{(M)}(x) \leq 1$. Furthermore, in the limit where $N \rightarrow \infty$ for a fixed, real $x$ and fixed non-negative integer  $M$, the following holds:
    \be\label{Thetalimit}
        \lim_{N\rightarrow \infty} \Theta_N^{(M)} (x) = \Theta[1-x] \ ,
    \ee
    where $\Theta[x]$ is the Heaviside step function. We will also need the following asymptotic formula, see e.g. \cite[Eq. (2.9)]{FyodorovCMP}: 
    \be\label{Eq:GammaEdgeAsymptotics}
        \lim_{N\rightarrow \infty} \frac{\Gamma\left(N-1,N+2\delta N^{1/2}\right)}{\Gamma\left( N-1 \right)} = \frac{1}{\sqrt{2\pi}} \int_{2\delta}^{\infty} dv \ \exp \left[ - \frac{v^2}{2} \right] = \frac{1}{2} \text{erfc}\left( \sqrt{2} \delta \right) \ .
    \ee
    For a fixed, finite $\delta$ we have from the definition of the incomplete $\Gamma$-function that:
    \be
        \lim_{N\rightarrow \infty} \frac{\Gamma\left(N-1,\delta \right)}{\Gamma\left( N-1 \right)} = \ e^{-\delta} \ \sum_{k=0}^{\infty} \frac{\delta^k}{k!} = 1 \ .
    \ee
    We will also use the following large $N$ asymptotic behaviour of the complementary error function $\text{erfc}(x)$,  see e.g. \cite[7.12.1]{NIST}
    \be\label{errorfuncexpansion}
        \text{erfc}\left(\sqrt{2} y \right) = \frac{e^{-2y^2}}{y \sqrt{2\pi}} \sum_{m=0}^{\infty} (-1)^m \frac{(2m-1)!!}{(4y^2)^m} \approx \frac{e^{-2y^2}}{y \sqrt{2\pi}} \ ,
    \ee
    in particular, the above implies that as $y\rightarrow \sqrt{N}y$ and $N \rightarrow \infty$, the following term vanishes
    \be\label{errorfuncres}
        \begin{split}
        e^{2y^2} \ \frac{1}{2y} \ \text{erfc}\left(\sqrt{2}y \right) &\rightarrow e^{2N y^2} \ \frac{1}{2\sqrt{N} y} \ \frac{e^{-2N y^2}}{\sqrt{N} y \sqrt{2\pi}} 
        = \frac{1}{2 N y^2 \sqrt{2\pi}} 
        \overset{N \gg 1}{\longrightarrow } 0 \ .
    \end{split}
    \ee

\end{rem}

\noindent
We now start with the asymptotic analysis in the bulk of the GinOE and give the proof of Corollary \ref{cor:GinOEbulkStrong}.

\begin{proof}
    Starting from Eq. \eqref{Eq:MainResGinOE},  we introduce the scaling $z\rightarrow \sqrt{N} w$, with $w= x+iy$ such that $|w| < 1$ and keeping $|y| > N^{1/2}$ to avoid the depletion regime. This transforms the 
    Eq. \eqref{Eq:MainResGinOE} to 
    \be
    \begin{split}
        \frac{1}{N} \mathcal{O}_N^{(\textup{GinOE},c)}(\sqrt{N} w) &= \frac{1}{\pi \ N } \ \left( 1 + \sqrt{\frac{\pi}{2}} \ \exp \left[ 2 \ N \ y^2 \right] \ \frac{1}{2 \ N \ \vert y \vert} \ \textup{erfc}\left(\sqrt{2} \ N \  \vert y \vert \right) \right) \\
        &\times N \ \bigg[  \ \frac{\Gamma\left(N-1,N \vert w \vert^2 \right)}{\Gamma\left( N-1 \right)}  \bigg[ 1 - \frac{1}{N} - \vert w \vert^2 \bigg] +  \frac{1}{N} \frac{N^{N-1} \ \vert w\vert^{2(N-1)}}{(N-2)!}  \ e^{-N \ \vert w \vert^2} \bigg] \ .
    \end{split}
    \ee
    The last term in the above expression vanishes, because $\vert w \vert <1$ and $N^{N-2}/(N-2)! \rightarrow 0$ as $N\rightarrow \infty$ via Stirling's formula. From Remark \ref{rem:erfcAsymptotic} we deduce that  the term containing the error function vanishes and  the ratio of $\Gamma$-functions can be identified with $\Theta^{(2)}_N(\vert z \vert^2)$, which leaves us with an expression of the form
    \be
    \begin{split}
        \frac{1}{N} \mathcal{O}_N^{(\textup{GinOE},c)}(\sqrt{N} w) &\rightarrow \frac{1}{\pi}  \ \Theta^{(2)}_N(\vert w \vert^2)  \bigg[ 1 - \frac{1}{N} - \vert w \vert^2 \bigg] \ .
    \end{split}
    \ee
    Taking the limit straightforwardly results in 
    \be
        \mathcal{O}_{\textup{bulk}}^{\textup{(GinOE,c)}}(w) \equiv \lim_{N\rightarrow \infty} \frac{1}{N} \ \mathcal{O}_N^{\textup{(GinOE,c)}}\left( \sqrt{N}w \right) 
        = \frac{1}{\pi} \left( 1 - \vert w \vert^2 \right) \ \Theta\left[ 1 - \vert w \vert^2 \right] \ ,
    \ee
     in full accordance with the result of Chalker \& Mehlig in the bulk of the GinUE \cite{CM}.
\end{proof}

\noindent
We now proceed with calculating the large $N$ asymptotic behaviour of the mean self-overlap at the edge, still keeping away from the depletion regime of the GinOE, providing a proof of Corollary \ref{cor:GinOEedgeStrong}.

\begin{proof}
    For the edge of the spectrum we consider the scaling
    \be
        z = x+iy = \left( \sqrt{N} + \eta \right) e^{i\theta} =  \left( \sqrt{N} + \eta \right) \cos \left( \theta \right) + i \left( \sqrt{N} + \eta \right) \sin \left( \theta \right) \ ,
    \ee
    which implies that the imaginary part $y$ is to be replaced by
    \be
        y = \left( \sqrt{N} + \eta \right) \sin \left( \theta \right) \overset{N\gg 1}{\approx}  \sqrt{N} \ \sin \left( \theta \right) \ ,
    \ee
    with the condition that $\vert \sin \left( \theta \right) \vert \sim \mathcal{O}(1)$  to avoid the depletion regime.  Together with Remark \ref{rem:erfcAsymptotic}, this implies immediately that the term containing the error function will again vanish in this scaling limit.  Also, as $N$ becomes large, the absolute value of $z$ becomes
    \be
        \vert z \vert^2 = \left( \sqrt{N} + \eta \right)^2 = N + 2\eta \ \sqrt{N} + \eta^2 \overset{N\gg 1}{\approx} N + 2\eta \ \sqrt{N} \ .
    \ee
    Recalling the finite $N$ result in Eq. \eqref{Eq:MainResGinOE},  the limiting mean self-overlap at the edge becomes
    \begingroup
    \allowdisplaybreaks
    \begin{align}
        \frac{1}{\sqrt{N}} \ \mathcal{O}_N^{(\textup{GinOE},c)}(z)  
        &= \frac{1}{\pi} \ \bigg[  \ \frac{\Gamma\left(N-1,N + 2\eta \ \sqrt{N} \right)}{\Gamma\left(N-1\right)}  \bigg( - \frac{1}{\sqrt{N}}- 2\eta \bigg) \\
        &+ \frac{1}{\sqrt{N}} \ \frac{1}{(N-2)!} \ \left( N + 2\eta \ \sqrt{N} \right)^{N-1}  \ e^{-N - 2\eta \ \sqrt{N}} \bigg] \nonumber \ .
    \end{align}
    \endgroup
    Now, by invoking Eq. \eqref{Eq:GammaEdgeAsymptotics} in Remark \ref{rem:erfcAsymptotic} we see that the first term tends to $- \eta \  \text{erfc}\left(\sqrt{2}\eta\right)$. In the second term, after using the Stirling approximation $(N-2)! \sim \sqrt{2\pi } \ e^{-N} \  N^{N-3/2}$, we take the limit, straightforwardly arriving at 
    \begin{equation}\label{Eq:Edge_asymptotic_formula}
        \lim_{N \to \infty} \frac{1}{\sqrt{N}} \frac{(N + 2 \eta \sqrt{N})^{N-1}}{(N-2)!} e^{-N - 2\eta \sqrt{N}} = \lim_{N \to \infty} \frac{1}{\sqrt{2\pi}} \left( 1 + \frac{2 \eta}{\sqrt{N}} \right)^{N-1} e^{-2\eta\sqrt{N} } = \frac{1}{\sqrt{2\pi}} e^{-2\eta^2} \ .
    \end{equation}
    Combining the two contributions thus gives the required result:
    \be
        \mathcal{O}_{\textup{edge}}^{\textup{(GinOE,c)}}(\eta) = \frac{1}{\pi} \left( \frac{1}{\sqrt{2\pi}} \ e^{-2\eta^2}  - \eta \ \textup{erfc}\left( \sqrt{2} \ \eta \right) \right) \ .
    \ee
\end{proof}

\newpage

\noindent
Finally we consider the depletion region covered by Corollary \ref{cor:GinOEdepletionStrong} and proceed to give its proof.

\begin{proof}
    We consider the limiting behaviour of the mean self-overlap of eigenvectors associated with eigenvalues of the form $z=x+i\xi$, where $\xi \sim \mathcal{O}(1)$ as $N \rightarrow \infty$. If $x\sim \mathcal{O}(1)$, i.e. we stay close to the origin as $N \rightarrow \infty$, then we have that $\vert z \vert^2 \sim \mathcal{O}(1)$. Recalling Eq. \eqref{Eq:MainResGinOE} and inserting into it that $z=x + i \xi$, we find that
    \be\label{Eq:proofCorGinOEdepletion}
    \begin{split}
        \mathcal{O}_N^{(\textup{GinOE},c)}(z) &= \frac{1}{\pi} \ \left( 1 + \sqrt{\frac{\pi}{2}} \ \exp \left[ 2\xi^2 \right] \ \frac{1}{2 \vert \xi \vert} \ \textup{erfc}\left(\sqrt{2} \ \vert \xi \vert \right) \right) \\
        &\times \bigg[  \ \frac{\Gamma\left(N-1, x^2 + \xi^2 \right)}{\Gamma\left(N-1 \right)}  \bigg[ N - 1 - x^2 - \xi^2 \bigg] +  \frac{\left( x^2 + \xi^2 \right)^{N-1}}{(N-2)!}  \ e^{-x^2 - \xi^2 } \bigg] \ .
    \end{split}
    \ee
    Multiplying with overall factor $1/N$, noticing that the final term vanishes as $N\to \infty$ and recalling Remark \ref{rem:erfcAsymptotic}, one arrives at
    \be
        \mathcal{O}^{\textup{(GinOE,c)}}_{\textup{depletion,origin}} (\xi) = \frac{1}{\pi} \left( 1+ \sqrt{\frac{\pi}{2}} \  \frac{1}{2 \vert \xi \vert } \ e^{2\xi^2} \ \textup{erfc}\left( \sqrt{2} \ \vert \xi \vert \right) \right)  \ ,
    \ee
    thus reproducing  Eq. \eqref{Eq:GinOEdepletionStrongRes}.\\

    \noindent
    If we scale the real part as $x= \sqrt{N}\delta$ instead, the limiting self-overlap acquires an additional contribution from the $\Gamma$-function terms in the second line in Eq. \eqref{Eq:proofCorGinOEdepletion}. This can be seen starting from the finite $N$ expression
    \be\label{Eq:proofCorGinOEdepletion2}
    \begin{split}
         \mathcal{O}_N^{(\textup{GinOE},c)}(z) &= \frac{1}{\pi} \ \left( 1 + \sqrt{\frac{\pi}{2}} \ \exp \left[ 2\xi^2 \right] \ \frac{1}{2 \vert \xi \vert} \ \textup{erfc}\left(\sqrt{2} \ \vert \xi \vert \right) \right) \\
        &\times \bigg[  \ \frac{\Gamma\left(N-1, N \delta^2 + \xi^2 \right)}{\Gamma\left(N-1 \right)}  \bigg[ N \left( 1 - \delta^2 \right) - 1 - \xi^2 \bigg] +  \frac{N^{N-1} \left( \delta^2 + \xi^2 \right)^{N-1}}{(N-2)!}  \ e^{-N \ \delta^2 - \xi^2 } \bigg] \ ,
    \end{split}
    \ee
    then utilising Remark \ref{rem:erfcAsymptotic}, which leads to an additional $\Theta\left[1- \delta^2 \right]$. Employing an overall scaling by $1/N$ to ensure a non-trivial limit completes the proof and we get Eq. \eqref{Eq:GinOEdepletionStrongRes2} for the strip in the depleted region, i.e.
    \be
        \mathcal{O}^{\textup{(GinOE,c)}}_{\textup{depletion,strip}} (\delta,\xi) =  \mathcal{O}^{\textup{(GinOE,c)}}_{\textup{depletion,origin}}(\xi) \left( 1- \delta^2 \right) \Theta\left[1- \delta^2 \right] \ .
    \ee
\end{proof}


\subsubsection*{Acknowledgements}

We would like to thank Gernot Akemann, Mario Kieburg and Wojciech Tarnowski for useful discussions. This research has been supported by the EPSRC Grant EP/V002473/1 “Random Hessians and Jacobians: theory and applications”.


\appendix


\section{Consistency Check: $\mathcal{O}(z)$ in the GinUE}\label{AppA}

\noindent
In the limit of $N \rightarrow \infty$ it is natural to expect the statistics of complex eigenvalues and eigenvectors  of the real and complex Ginibre ensembles to match in the bulk and at the edge of the disk. For the GinUE, the finite $N$ joint probability density function of an eigenvalue $z=x+iy$, with $y\neq 0$, and the associated self-overlap of the eigenvectors, $\mathcal{O}$  has been derived in a rather complicated form in \cite{FyodorovCMP}, which reads: 
\begin{align}
    \mathcal{P}^{\text{(GinUE,c)}}_N(\mathcal{O}, z) = \frac{1}{\pi \Gamma(N) \Gamma(N-1)} \frac{e^{\frac{|z|^2}{\mathcal{O}}}}{\mathcal{O}^3} \left( \frac{\mathcal{O} - 1}{\mathcal{O}} \right)^{N-2} \left[ D_1^{(N)} + |z|^2 \frac{D_2^{(N)}}{\mathcal{O}} + |z|^4 \frac{d_1^{(N)}}{\mathcal{O}^2}  \right] \ ,
\end{align} 
where
\begingroup
\allowdisplaybreaks
\begin{align}
    D_1^{(N)} =& ~ |z|^4(N-1)(N-2)d_1^{(N-1)} + \left[ (N-1)N - 2|z|^2(N + |z|^2) \right]d_1^{(N)} \\
    \nonumber & ~ -|z|^2(N-2)(N - |z|^2) d_2^{(N-1)} + |z|^2 d_2^{(N)} \\
    D_2^{(N)} =& ~ 2N d_1^{(N)} - |z|^2 (N - 2) d_2^{(N - 1)} \\
    d_1^{(N)} =& ~ \Gamma\left( N - 1, |z|^2 \right) \Gamma\left( N + 1, |z|^2 \right) - \Gamma\left( N , |z|^2 \right) \Gamma\left( N, |z|^2 \right) \\
    d_2^{(N)} =& ~ \Gamma\left( N - 1, |z|^2 \right) \Gamma\left( N + 2, |z|^2 \right) - \Gamma\left( N , |z|^2 \right) \Gamma\left( N + 1, |z|^2 \right)  \ .
\end{align}	
\endgroup
For completeness, we check below that its first moment produces the known mean self-overlap at finite $N$.\\

\noindent
By definition, the mean self-overlap is the first moment of the above JPDF and therefore can be obtained by integration. Hence, we have
\begin{equation}
    \mathcal{O}^{\text{(GinUE,c)}}_{N}(z)  = \int_1^{\infty} d \mathcal{O} \  \mathcal{O} \int_{\mathbb{C}} d^2 z \ \delta(z - \Tilde{z}) \  \mathcal{P}^{\text{(GinUE,c)}}_N (\mathcal{O}, z) \ ,
\end{equation}
which when written explicitly has the following form:
\begin{align}
     \mathcal{O}^{\text{(GinUE,c)}}_{N}(z) =& \frac{1}{\pi \Gamma(N) \Gamma(N-1)} \int_1^\infty d \mathcal{O} \  \frac{e^{\frac{|z|^2}{\mathcal{O}}}}{\mathcal{O}^2} \left( \frac{\mathcal{O} - 1}{\mathcal{O}} \right)^{N-2} \left[ D_1^{(N)} + |z|^2 \frac{D_2^{(N)}}{\mathcal{O}} + |z|^4 \frac{d_1^{(N)}}{\mathcal{O}^2}  \right] \\
    =& \frac{1}{\pi \Gamma(N) \Gamma(N-1)} \left[ D_1^{(N)} I_1 + |z|^2 D_2^{(N)} I_2 + |z|^4 d_1^{(N)} I_3 \right] \ ,
    \label{eq:EGinUE_I123}
\end{align}
where the integrals $I_1$, $I_2$ and $I_3$ are defined below and can be readily evaluated in terms of the incomplete $\Gamma$-function:
\begingroup
\allowdisplaybreaks
\begin{align}
    I_1 =& \int_1^\infty d \mathcal{O} \  e^{\frac{|z|^2}{\mathcal{O}}} \frac{ \left(\mathcal{O} - 1\right)^{N-2} }{\mathcal{O}^N} = \frac{e^{|z|^2}}{|z|^{2(N-1)}} \left( \Gamma(N-1) - \Gamma(N-1, |z|^2) \right) \\
    I_2 =& \int_1^\infty d \mathcal{O} \  e^{\frac{|z|^2}{\mathcal{O}}} \frac{ \left(\mathcal{O} - 1\right)^{N-2} }{\mathcal{O}^{N+1}} = \frac{\left( |z|^{2N} - e^{|z|^2}(N - |z|^2 - 1) \left[ \Gamma(N) - \Gamma(N, |z|^2) \right] \right) }{(N-1) |z|^{2N}} \\
    I_3 =& \int_1^\infty d \mathcal{O} \  e^{\frac{|z|^2}{\mathcal{O}}} \frac{ \left(\mathcal{O} - 1\right)^{N-2} }{\mathcal{O}^{N+2}} = \frac{\Gamma(N-1)}{|z|^{2(N+1)}\Gamma(N + 2)} \bigg( (N + 2 - N^2)|z|^{2(N+1)} + (N+1) |z|^{2(N+2)} \\
    & \nonumber + e^{|z|^2} \left( |z|^4 - 2(N-1)|z|^2 + N(N - 1) \right) \left[ \Gamma(N + 2) - (N+1) \Gamma(N+ 1, |z|^2) \right]  \bigg)  \ .
\end{align}
\endgroup
To combine all the above results we utilised Mathematica software to manipulate the resulting expressions, finally finding that
\begin{equation}\label{Eq:App_Overlap_finiteN_GinUE}
    \mathcal{O}^{(\GinUE,c)}_N(z) = \frac{1}{\pi} \left[ \frac{\Gamma(N, |z|^2)}{(N-1)!}(N - |z|^2) + \frac{|z|^{2N}}{(N-1)!} e^{-|z|^2} \right] \ ,
\end{equation}
which can be shown to be equivalent to the forms appearing in the literature, see e.g. \cite{WS}. The corresponding edge and bulk scaling limits can be straightforwardly recovered using essentially the same steps as employed for the GinOE case.

\end{document}